\documentclass[11pt,onecolumn]{article}
\usepackage[utf8]{inputenc}

\usepackage[a4paper]{geometry}

\usepackage
[
colorlinks=true,
linkcolor=blue,
anchorcolor=blue,
citecolor=blue,
urlcolor=blue,
plainpages=false,
pdfpagelabels
]{hyperref}

\usepackage{amsmath}
\usepackage{amsthm}

\usepackage{thmtools}
\usepackage{enumerate}

\usepackage{amssymb}
\usepackage{microtype,t1enc,txfonts}
\usepackage{multirow}
\usepackage{xspace}
\usepackage{graphicx}
\usepackage{epstopdf}
\usepackage{ifpdf}
\usepackage{url,hyperref}
\usepackage{latexsym}
\usepackage{euscript}
\usepackage{xspace}
\usepackage{color}
\usepackage{makeidx}
\usepackage{picins}
\usepackage{wrapfig}
\usepackage{xypic}
\usepackage{tikz,tikz-cd}
\usetikzlibrary{matrix,arrows}
\usepackage[capitalize]{cleveref}

\usepackage[numbers,sort,compress]{natbib}
\xyoption{all}

\long\def\remove#1{}

\newtheorem{theorem}{Theorem}[section] 
\newtheorem{lemma}[theorem]{Lemma}
\newtheorem{prop}[theorem]{Proposition}

\newtheorem{obs}[theorem]{Observation}

\newtheorem{cor}[theorem]{Corollary}

\newtheorem{mydef}[theorem]{Definition}

\usepackage{todo}

\makeatletter
\renewcommand{\@tododisplay}[1]{%
}
\makeatother

\newcommand {\mm}[1] {\ifmmode{#1}\else{\mbox{\(#1\)}}\fi}

\newcommand{\eps}{{\varepsilon}}
\newcommand{\reals}	{{\rm I\!\hspace{-0.025em} R}}

\newcommand{\etal}      {et al.\@\xspace}

\newcommand{\rg}		{ {R}}
\newcommand{\HH}		{{{H}}}
\newcommand{\ZZ}		{{{Z}}}
\newcommand{\XX}		{{{X}}}
\newcommand{\YY}		{{{Y}}}
\newcommand{\CC}	{{{C}}}
\newcommand{\DD}		{{d_{FD}}} 
\newcommand{\Z}		{{\mathbb{Z}}}
\newcommand{\R}		{{\mathbb{R}}}
\newcommand{\surReeb}		{{\mu}}
\DeclareMathOperator\Dg{Dg} 
\DeclareMathOperator\eDg{ExDg} 
\DeclareMathOperator\height{height}
\DeclareMathOperator\range{range}

\definecolor{red}{rgb}{0.6, 0, 0}
\definecolor{blue}{rgb}{0, 0, 0.6}
\definecolor{green}{rgb}{0.16, 0.435, 0.16}

\renewcommand{\paragraph}[1]	{{\vspace*{0.1in}\noindent {\bf #1.~}}}

\newcommand{\GHlike}	{functional distortion} 

\newcommand{\dreeb}		{d}
\newcommand{\leftmap}		{\phi}
\newcommand{\rightmap}		{\psi}

\newcommand{\optd}			{\delta}

\newcommand{\mycanonical}	{{thin\xspace}}
\newcommand{\Mycanonical}	{{Thin\xspace}}
\newcommand{\frank}		{n}
\newcommand{\grank}		{m}
\newcommand{\loopone}		{\gamma}
\newcommand{\looptwo}		{\zeta}
\newcommand{\aloop}		{\gamma}

\newcommand{\gset}			{{\mathcal{G}}}
\newcommand{\smallH}		{{{Z}_1^{2\delta}}}

\DeclareMathOperator\domcycle{dom}

\newcommand{\bp}			{{b}}
\newcommand{\tp}			{{d}}
\newcommand{\onerg}		{{{R}}} 

\newcommand{\newrg}		{{\widetilde{\onerg}}}
\DeclareMathOperator\im{im}
\DeclareMathOperator\id{id}

\newcommand{\Fmatrix}		{\Phi}

\newcommand{\augGH}		{{functional GH\xspace}}
\newcommand{\mycycle}		{{cycle\xspace}}
\newcommand{\mydecomposition}	{{\mycanonical{} basis decomposition\xspace{}}}

\newcommand{\ignore}[1]{}

\newcommand{\fGH}		{{functional distortion\xspace}} 
\newcommand{\treeone}	{T_f}
\newcommand{\treetwo}	{T_g}
\newcommand{\Dint}		{d_I}
\newcommand{\Dfgh}		{d_{fGH}}
\newcommand{\leftintmap}	{\alpha^\eps}
\newcommand{\rightintmap}	{\beta^\eps}
\newcommand{\ishift}		{{\mathsf{i}}}
\newcommand{\jshift}		{{\mathsf{j}}}

\newcommand{\rightmapit}		{\phi_{\leftarrow}}
\newcommand{\optleftmapit}	{\phi^*} 
\newcommand{\optrightmapit}	{\psi^*}
\newcommand{\leftintmapdelta}	{\alpha^\delta}
\newcommand{\rightintmapdelta}	{\beta^\delta}

\begin{document}

\title{Measuring Distance between Reeb Graphs}
\author{Ulrich Bauer\thanks{Department of Mathematics, Technical University of Munich (TUM), D-85747 Garching. \url{http://ulrich-bauer.org}} \and Xiaoyin Ge\thanks{Computer Science and Engineering Department, The Ohio State University, Columbus, OH 43221. Emails: \href{mailto:gex@cse.ohio-state.edu}{\tt gex}, {\href{mailto:yusu@cse.ohio-state.edu}{\tt yusu@cse.ohio-state.edu}}.} \and Yusu Wang$^\dagger$}

\maketitle

\begin{abstract}
We propose a metric for Reeb graphs, called the \GHlike{} distance. 
Under this distance, the Reeb graph is stable against small changes of input functions. 
At the same time, it remains discriminative at differentiating input functions. 
In particular, the main result is that the \GHlike{} distance between two Reeb graphs is bounded from below by the bottleneck distance between both the ordinary and extended persistence diagrams for appropriate dimensions. %

As an application of our results, we analyze a natural simplification scheme for Reeb graphs, and show that persistent features in Reeb graph remains persistent under simplification. Understanding the stability of important features of the Reeb graph under simplification is an interesting problem on its own right, and critical to the practical usage of Reeb graphs. 
\end{abstract}

\section{Introduction}
\label{sec:intro}

One of the prevailing ideas in geometric and topological data analysis is to provide descriptors that encode
useful information about hidden objects from observed data. The Reeb graph is one such descriptor. 
Specifically, given a continuous function $f: X \rightarrow \reals$ defined on a domain $X$, the level set of $f$ at value $a$ is the set $f^{-1}(a) = \{ x \in X \mid f(x) = a \}$. 
As the scalar value $a$ increases,
connected components appear, disappear, split and merge in the level set, 
and the Reeb graph of $f$ tracks such changes. 
It provides a simple yet meaningful abstraction of the input domain. 
The concept behind the Reeb graph was first introduced by G.~Reeb 
in \cite{Reeb46} for Morse functions on manifolds; 
the term \emph{Reeb graph} was coined by R.~Thom. The first use of Reeb graphs for visualization applications can be found in work on shape understanding by Shinagawa \etal{} \cite{SK91}. 
Since then, it has been used in a variety of applications in graphics and visualization, e.g,  \cite{HA03,HSKK01,NBPF11,SK91,Tie08,WHDS04}; 
also see \cite{BGSF08} for a survey. 

The Reeb graph can be computed efficiently in $O(m\log m)$ time for a piecewise-linear function 
defined on an arbitrary simplicial complex domain with $m$ vertices, edges and triangles \cite{Salman12} (a randomized algorithm was given in \cite{HWW10}). 
This is in contrast to, for example, the $O(m^3)$ time (or matrix multiplication time) needed to compute even just the first-dimensional homology information for the same simplicial complex. 
The Reeb graph of a scalar field on a manifold can also be approximated from a point sample efficiently and with theoretical guarantees \cite{DW13}. 
It encodes meaningful information on the input scalar field, in particular the so-called one-dimensional \emph{vertical homology group} \cite{DW13}.  
Being a graph structure, the Reeb graph is simple to represent and manipulate. 
These properties make the Reeb graph appealing for analyzing high-dimensional point data. 
For example, a generalization of the Reeb graph is proposed in \cite{SMC07} for analyzing high dimensional data, and in \cite{GSBW11}, the Reeb graph is used to recover a hidden geometric graph from its point samples. 
Very recently in \cite{CS14}, it is shown that a certain Reeb graph can reconstruct a metric graph
with respect to the Gromov-Hausdorff distance. 

Given the popularity of the Reeb graph in 
data analysis, it is important to understand its stability and robustness with respect to changes in the input function (both in function values and in the domain). 
To measure the stability, we first need to define a distance between two Reeb graphs. 
Furthermore, an important application of the Reeb graph is to provide a descriptive summary of the function. %
Again, a central problem involved is to have a meaningful distance between Reeb graphs. 

In the special case of Reeb graphs of functions on curves, similar results were obtained in \cite{DiFabio2012Stability} using an editing distance on Reeb graphs, and this approach is being extended to surfaces by the same authors.
Recently, Morozov \etal{} proposed the \emph{interleaving distance} for \emph{merge trees}, based on the concept of an interleaving \cite{CCG09}, and obtained similar upper and lower bounds relating this distance to \emph{ordinary} persistence diagrams \cite{MBW13}. 
Here, the merge trees are variants of the loop-free Reeb graphs (\emph{contour trees}). 
However, it is not clear how to generalize these results to Reeb graphs containing loops, an important family of features of the Reeb graph. 
Another distance based on the \emph{branch decomposition} of merge trees was proposed in \cite{BYM13}, together with a polynomial time algorithm to compute it. This distance, however, is not stable with respect to changes in the function and also does not generalize beyond trees. 

Recently, de Silva et al. introduced the \emph{interleaving distance} for Reeb graphs, which is defined at the algebraic topology level, utilizing the equivalence between Reeb graphs and a particular class of cosheaves \cite{DMP14}. In a previous conference paper \cite{Bauer2014}, we introduced the functional distortion distance to be described in the current full version. Notably, it has been shown very recently in \cite{BMW15} that these two definitions of distances between Reeb graphs are strongly equivalent, in the sense that they are within constant factor of each other.

\paragraph{Our work} 
In this paper, we propose a metric for Reeb graphs, called the \emph{\GHlike{} distance}, drawing intuition from the Gromov-Hausdorff distance for measuring metric distortion. 
Under this distance, the Reeb graph is stable against perturbations of the input function; at the same time, it retains a certain ability to discriminate between different functions (these statements will be made precise in \cref{sec:properties}). 
In particular, the main result is that the \GHlike{} distance between two Reeb graphs is bounded from below by (and thus more discriminative than) the \emph{bottleneck distance} between the persistence diagrams \emph{of the Reeb graphs}.
On the other hand, the \GHlike{} distance yields the same type of sup norm stability that persistence diagrams enjoy \cite{CEH07,CCG09,CSGO12,Bauer2013Induced}. The persistence diagram has been a popular topological summary of shapes and functions, and the bottleneck distance is introduced in \cite{CEH07} as a natural distance for persistence diagrams. 
However, as the simple example in \cref{fig:twotypes} (a) shows, the Reeb graph can be strictly more discriminative than the persistence diagram of dimension 0. 

In Section \ref{sec:GHrelation}, we show the relation between our functional distortion distance to a functional-version of the Gromov-Hausdorff distance. 
In Section \ref{sec:interleaving}, we show that, when applied to merge trees, our functional distortion distance is equivalent to the interleaving distance proposed by Morozov et al. \cite{MBW13}. 

Finally, as an application of our results, we show in \cref{sec:simp} that persistent features of the Reeb graph remain persistent under a certain natural simplification strategy of the Reeb graph. 
Understanding the stability of Reeb graph features under simplification is an interesting problem on its own right: In practice, one often collapses small branches and loops in the Reeb graph to remove noise; see, e.g., \cite{DN12,GSBW11,PSBM07}. It is crucial that by collapsing a collection of small features, there is no cascading effect that causes larger features to be destroyed, and our results confirm that this is indeed the case.

\section{Preliminaries and Problem Definition}
\label{sec:background}

\paragraph{Reeb graphs}
Given a continuous function $f: \XX \rightarrow \reals$ on a finitely triangulable 
topological space~$\XX$, for each $\alpha \in \reals$, the set $f^{-1}(\alpha)=\{x\in \XX: f(x)=\alpha \}$ is called a \emph{level set} of $f$. A level set may consist of several connected components. 
We define an equivalence relation~$\sim$~on~$\XX$ such that $x \sim y$ iff $ f(x) = f(y) = \alpha$ and $x$ is connected to~$y$ in~$f^{-1}(\alpha)$. 
\begin{figure}[h]
\centering\includegraphics[width=7cm]{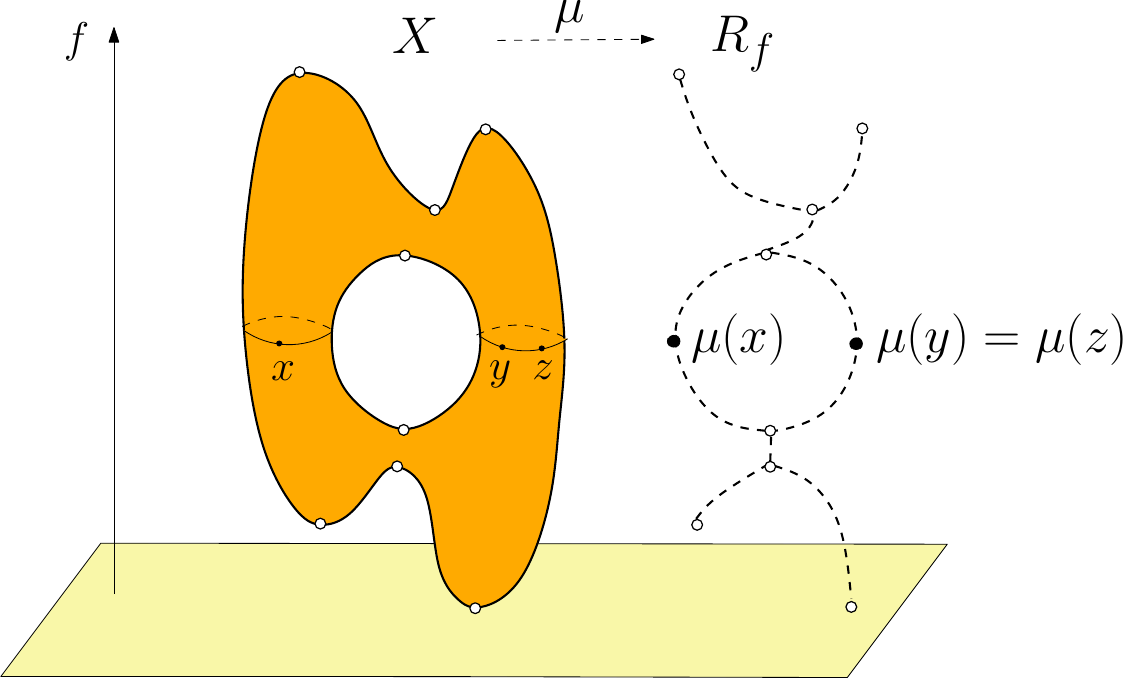}
\end{figure}
The \emph{Reeb space} of the function $f: \XX \rightarrow \reals$, denoted by $\rg_f$, 
is the quotient space $\XX/{\sim}$, i.e., the set of 
equivalent classes equipped with the quotient topology induced by
the quotient map $\surReeb: \XX \rightarrow \rg_f$. 
Under appropriate regularity assumptions
(to be made precise later),
$\rg_f$ has the structure of a finite $1$-dimensional regular CW complex, and we call it a \emph{Reeb graph}.
Throughout this paper, we tacitly assume that all mentioned connected components are also path-connected. 

The input function~$f:\XX\to\reals$ also induces a 
continuous function $\tilde{f}: \rg_f \rightarrow \reals$ defined as
$\tilde{f} (z) = f(x)$ for any preimage $x \in \surReeb^{-1}(z)$ of $z$. 
To simplify notation, we often write $f(z)$ instead of $\tilde f(z)$ for $z \in \rg_f$ when there is no ambiguity, and use $\tilde f$ mostly to emphasize the different domains of the functions.
In all illustrations of this paper, we plot the Reeb graph with 
the vertical coordinate of a point $z$ corresponding to the function value~$f(z)$.

Given a point $x \in \rg_f$, we use the term \emph{up-degree} (resp.~\emph{down-degree}) of $x$ to denote the number of branches (1-cells) incident to $x$ that have higher (resp.~lower) values of $f$ than $x$. 
A point is \emph{regular} if both of its up-degree and down-degree equal to 1, and \emph{critical} otherwise. 
A critical point is a minimum (maximum) if it has down-degree 0 (up-degree 0), and a down-fork (up-fork) 
if it has down-degree (up-degree) larger than~$1$. 
A critical point can be degenerate, having more than one types of criticality. %
From now on, we use the term \emph{node} to refer to a critical point in the Reeb graph. 
For simplicity of exposition, we assume that all nodes of the Reeb graph have distinct $\tilde f$ function values. 
Note that because of the monotonicity of $\tilde f$ at regular points, the Reeb graph together with its associated function is completely described, up to homeomorphisms preserving the function, by the function values on the nodes. 

\paragraph{Persistent homology and persistence diagrams} 
The notion of persistence was originally introduced by Edelsbrunner \etal{} in \cite{ELZ02}. There has since been a great amount of development both in theory and in applications; see, e.g., \cite{ZC05,CD08,CSGO12,Bauer2013Induced}. 
This paper does not concern the theory of persistence, hence we only provide a simple description so as to introduce the notion of \emph{persistence diagrams}, which will be used later. We refer the readers to \cite{Hatcher2002Algebraic} for a detailed treatment of homology groups in general and to \cite{EH09} for persistent homology. 

Given a continuous function $f: \XX \rightarrow \reals$ defined on a finitely triangulable 
topological space $\XX$, we call $\XX_{\leq a} = \{ x\in \XX \mid f(x) \leq a \}$
a \emph{sublevel set} of $f$. Let $\HH_p(\YY)$ denote the $p$-th homology group of a triangulable topological space~$\YY$. Recall that a triangulation gives a CW structure and singular, simplicial, and cellular homology are isomorphic (see~\cite{Hatcher2002Algebraic} for details).
In this paper, we always consider homology with coefficients in $\Z_2$, so $\HH_p(\YY)$ is a vector space. 
We now investigate the changes of $\HH_p(\XX_{\leq a})$ for increasing values of $a$. 
Throughout this paper, we will assume that $f$ is \emph{tame} in the following sense: there is a finite partition $-\infty = a_0 < \min f = a_1 < \dots < a_N = \max f < \infty = a_{N+1}$ such that for all $i<n$ and $s,t \in [a_i,a_{i+1})$ with $s<t$, the homomorphism $\HH_p (\XX_{\leq s}) \to \HH_p(\XX_{\leq t})$ induced by the inclusion $\XX_{\leq s} \hookrightarrow \XX_{\leq t}$ is an isomorphism, and similarly, for all $s,t \in (a_i,a_{i+1}]$ with $s<t$, the homomorphism $\HH_p (\XX_{\geq t}) \to \HH_p(\XX_{\geq s})$ induced by the inclusion $\XX_{\geq t} \hookrightarrow \XX_{\geq s}$ is an isomorphism. 
Moreover, $\HH_p (\XX_{\leq a_i})<\infty$ for all $i$. This implies that $\rg_f$ is a Reeb graph. 
We call $a_i$ a \emph{homologically critical level} of $f$.

Consider the following sequence of vector spaces, %
\begin{equation}
0 = \HH_p (\XX_{\leq a_0}) \to \HH_p(\XX_{\leq  a_1}) \to \cdots \to \HH_p(\XX_{\leq a_N}) = \HH_p(\XX), 
\label{eqn:tradseq}
\end{equation}
where each homomorphism $\mu_i^j: \HH_p (\XX_{\leq a_i}) \to \HH_p (\XX_{\leq  a_j})$ is induced by the canonical inclusion $\XX_{\leq a_i} \hookrightarrow \XX_{\leq a_j}$.

A homology class $h$ is created at $a_i$ if \[h \in \HH_p(\XX_{\leq a_i})\text{ but }h \not\in \im \mu_{i-1}^i.\]
It is destroyed at $a_j$ if \[\mu_i^{j-1}(h)\notin\im \mu_{i-1}^{j-1}\text{ but }\mu_i^j(h)\in\im \mu_{i-1}^j.\]
Persistent homology records such birth and death events. %
In particular, the $p$-th \emph{ordinary persistence diagram} of~$f$, denoted by $\Dg_p(f)$, is a multiset of pairs $(b, d)$ corresponding` to the birth value $b$ and death value $d$ of some $p$-dimensional homology class. 
See Figure~\ref{fig:twotypes}~(c) for an example of the $0$-th persistence diagram. 
(We note that this is only an intuitive and informal introduction of the persistence diagram; see \cite{EH09,ZC05} for a more formal treatment.) 
\begin{figure*}[tbp]
\begin{center}
\begin{tabular}{cccccc}
\includegraphics[height=2.8cm]{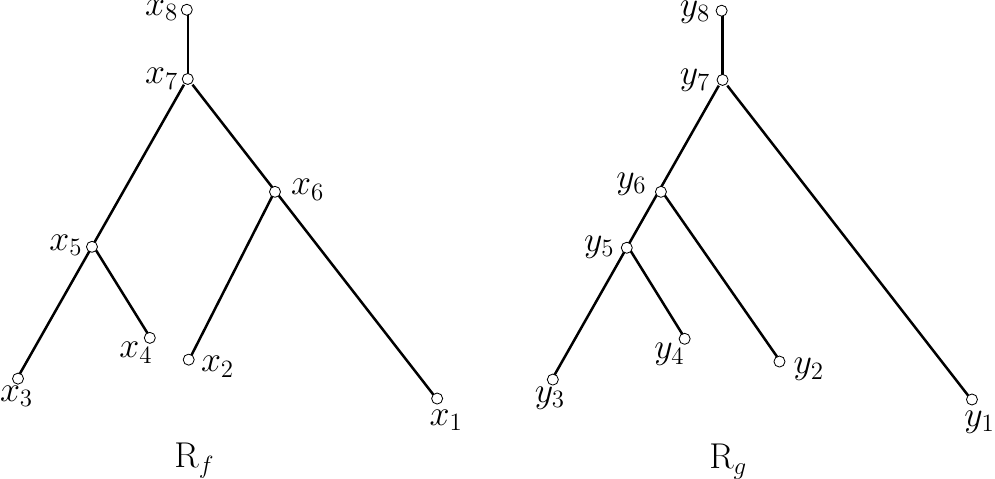} &
\includegraphics[height=3.3cm]{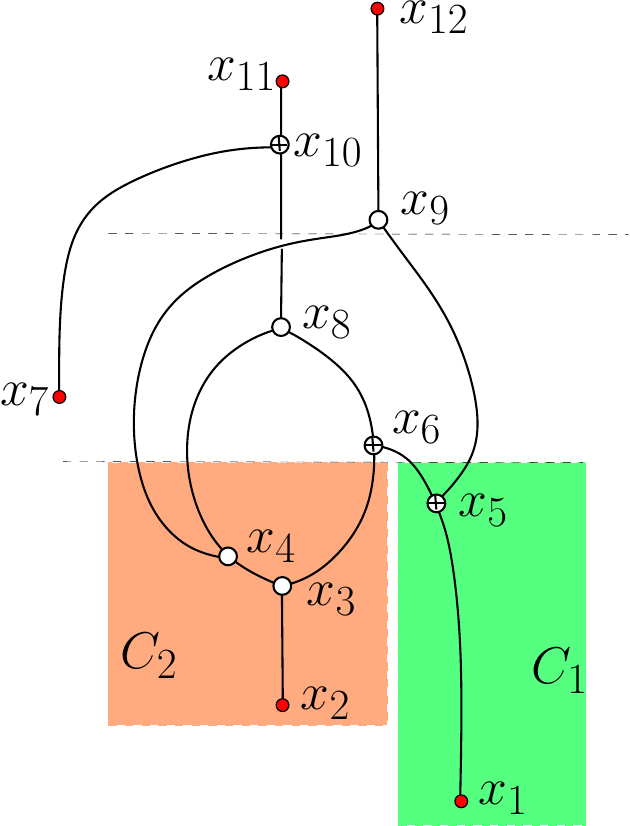} & 
\includegraphics[height=2.9cm]{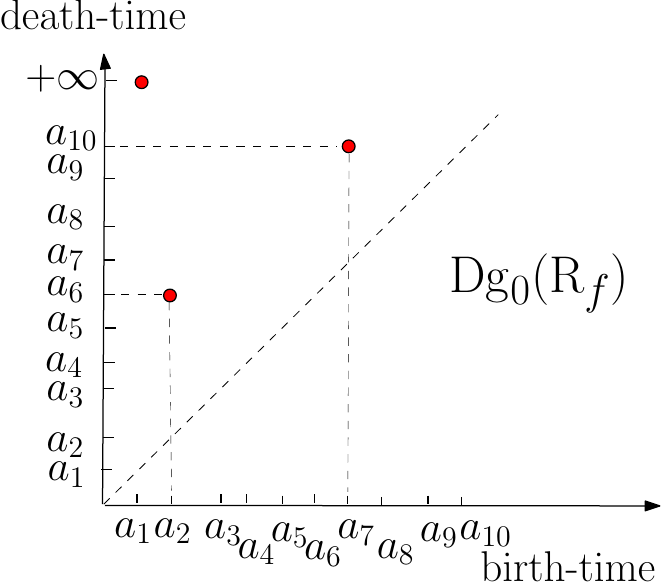} &
\includegraphics[height=2.9cm]{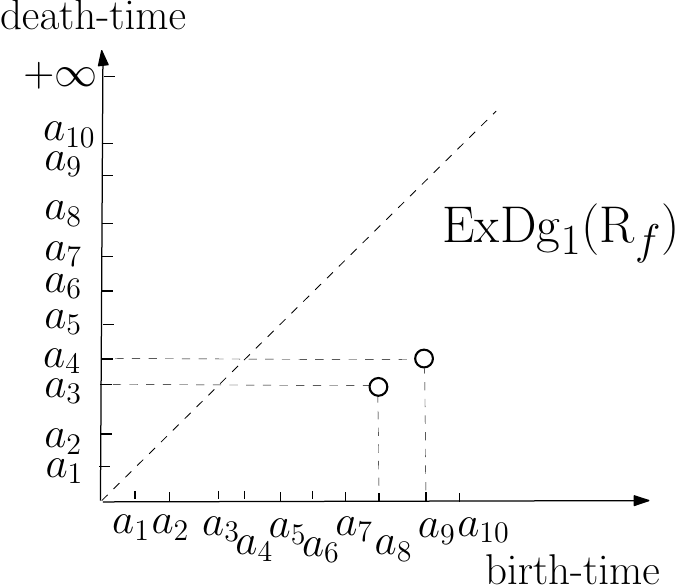} \\
(a) &  (b) & (c) & (d) 
\end{tabular}
\end{center}
\vspace*{-0.2in}\caption{{\small (a) The height functions on the two trees have the same persistence diagrams (thus the bottleneck distance between their persistence diagrams is 0), but their tree structures are different. The \GHlike{} distance will differentiate these two cases. 
In (b), solid dots are minimum and maximum, empty dots are essential forks, and crossed-dots are ordinary forks. The ordinary fork $x_6$ merges components $C_1$ and $C_2$ in the sublevel set below it, represented by minima $x_1$ and $x_2$ respectively. 
The resulting critical pair $(x_2, x_6)$ gives rise to the point $(a_2, a_6)$ in $\Dg_0(\rg_f)$ in (c), where $a_i = f(x_i)$ for $i \in [1, 12]$. 
The essential fork $x_9$ is paired with the up-fork $x_4$, corresponding to the \mycanonical{} loop $x_4 x_8 x_6 x_5 x_9 x_4$ created at $x_9$. This gives rise to the point $(a_4,a_9)$ in the extended persistence diagram $\eDg_1(\rg_f)$ in (d).  }
\label{fig:twotypes}}
\end{figure*}

In general, since $\HH_p (\XX)$ may not be trivial, any nontrivial homology class of $\HH_p (\XX)$, referred to as an \emph{essential homology class}, will never die during the sequence in \cref{eqn:tradseq}. 
For example, there is a point $(a_1, \infty)$ in \cref{fig:twotypes} (b) indicating a $0$-dimensional homology class that was created at $a_1$ but never dies. 
By appending a sequence of relative homology groups to \cref{eqn:tradseq}, we obtain a pairing of the essential homology classes (i.e., homology classes of $\HH_p(\XX)$): 
\begin{multline}
0 = \HH_p (\XX_{\leq a_0}) \to \cdots \to \HH_p(\XX_{\leq a_N}) = \HH_p(\XX) = \\ \HH_p(\XX, \XX_{\geq a_N}) \to \HH_p (\XX, \XX_{\geq a_{N-1}}) \to \cdots  \to \HH_p (\XX, \XX_{\geq a_0}) = 0.  
\label{eqn:extseq}
\end{multline}
Here $\XX_{\geq a}$ denotes the \emph{superlevel set} $\XX_{\geq a} = \{ x\in \XX \mid f(x) \ge a \}$. 
Since the last vector space $\HH_p(\XX, \XX_{\geq a_0}) = 0$, each essential homology class will necessarily die in the \emph{relative part} of the above sequence at some relative homology group $\HH_p(\XX, \XX_{\geq a_j})$. 
We refer to the multiset of points encoding the birth and death time of $p$th homology classes created in the ordinary part and destroyed in the relative part of the sequence in \cref{eqn:extseq} as the \emph{$p$th extended persistence diagram} of $f$, denoted by $\eDg_p(f)$. In particular, for each point $(b, d)$ in $\eDg_p(f)$ there is a (essential) homology class in $\HH_p(\XX)$ that is born in $\HH_p(\XX_{\leq b})$ 
and dies at $\HH_p(\XX, \XX_{\geq d})$. %
See \cref{fig:twotypes} (d) for an example; note that the birth time is larger than or equal to death time in the extended persistence diagram. 

\paragraph{Reeb graphs and persistent homology}
There is a natural way to define and quantify features of the Reeb graph, which turns out to be consistent with the information encoded in the diagrams 
$\Dg_0(\rg_f)$ and $\eDg_1(\rg_f)$ of the function $\tilde f: \rg_f \rightarrow \reals$. 
Since $\rg_f$ is a graph, we only need to consider persistent homology in dimensions 0 and 1. 
We provide an intuitive treatment below. 
For simplicity of exposition, we assume that all nodes have different function values and are either a minimum, a maximum, a down-fork with down-degree 2, or an up-fork with up-degree 2, noting that these assumptions hold in the generic case.

Imagine that we sweep through $\rg_f$ in increasing values of $a$ and inspect changes in $\HH_0( (\rg_f)_{\leq a})$. %
New components in the \emph{sublevel sets} are created at minima of $\rg_f$. 
For any value $a$, associate each component $C$ in the sublevel set of $(\rg_f)_{\leq a}$ with the lowest local minimum $m$ contained in $C$: intuitively, $C$ is created at $m$. 

Consider a down-fork node $s$ with $a = f(s)$. %
If the two lower branches are contained in different connected components $C_1$ and $C_2$ of the open sublevel set $(\rg_f)_{<a}$, for reasons that will become obvious soon we call~$s$ an \emph{ordinary fork}; otherwise, it is an \emph{essential fork}. 
Let $x_1$ and $x_2$ be the global minimum of $C_1$ and $C_2$, respectively. 
Assume that $f(x_1)<f(x_2)$.
Then the homology class $[x_2+x_1]$ is created at $f(x_2)$ and dies at $f(s)$, giving rise to a unique point $(f(x_2), f(s))$ in the $0$-th ordinary persistence diagram $\Dg_0(\rg_f)$. 
Indeed, 
there is a one-to-one correspondence between the set of 
such pairs of minima and ordinary down-forks
and points in the $0$th persistence diagram $\Dg_0(\rg_f)$ with finite coordinates; see \cref{fig:twotypes} (b) and (c). 
A symmetric procedure with $-f$ will produce pairs of maxima and ordinary up-forks, corresponding to points in the $0$th persistence diagram $\Dg_0(\rg_{-f})$. 
Together, these pairs capture the \emph{branching features} of a Reeb graph. 

If, on the other hand, the two lower branches of $s$ are connected in the 
sublevel set, we call $s$ an \emph{essential fork}; see \cref{fig:twotypes} (b) and (d). 
In this case, some cellular 1-cycle in the sublevel set $(\rg_f)_{\leq a}$ is born at $a$. 
Since $\rg_f$ is a graph, this cycles is non-trivial in $\rg_f$, and their corresponding homology classes will not be destroyed in ordinary persistent homology. 
Consider the unique cycle $\gamma$ with largest minimum value of $f$ among all cycles born at~$a$ and corresponding to an embedded loop in $\rg_f$.
Let $s'$ be the point achieving the minimum on $\gamma$. Then the cycle $\gamma$ is created at $f(s)$ during the ordinary sequence of \cref{eqn:extseq}, and killed at time $f(s')$ in the extended part, giving rise to a unique point $(\tilde f(s'), \tilde f(s))$ in the $1$st extended persistence diagram of $\tilde f$. It turns out that $s'$ is necessarily an essential up-fork \cite{AEHW06}, and we call such a pair $(s', s)$ an \emph{essential pair}. 
Indeed, 
the collection of essential pairs has a one-to-one correspondence to points in $\eDg_1(\rg_f)$. (The extended persistence diagram $\eDg_1(\rg_{-f})$ is the reflection of $\eDg_1(\rg_f)$ and thus encodes the same information as $\eDg_1(\rg_f)$.) 
These essential pairs capture the \emph{\mycycle{} features} of a Reeb graph. 

In short, the branching features and \mycycle{} features of a Reeb graph give rise to points in the $0$th ordinary and $1$st extended persistence diagrams, respectively. 
However, the persistence diagram captures only the lifetime of features, but not how these features are connected; see \cref{fig:twotypes} (a). In this paper we aim to develop a way of measuring distance between Reeb graphs which also takes into account the graph structure. 

\section{A Metric on Reeb Graphs}
\label{sec:metric}

Throughout this paper, by a \emph{distance} we will mean an extended pseudometric, i.e., a binary symmetric function $d$ with values in $[0,\infty]$ that satisfies $d(x,x)=0$ and $d(x,z)\leq d(x,y)+d(y,z)$.  From now on, consider two Reeb graphs $\rg_f$ and $\rg_g$, 
generated by tame functions $f: \XX \rightarrow \reals$ and $g: \YY \rightarrow \reals$. 
While topologically each Reeb graph is simply a 1-dimensional regular CW complex, it is important to note that it also has a function associated with it (induced from the input scalar field). 
Hence the distance should depend on both the graph structures and the functions $\tilde f$ and $\tilde g$. 
Approaching the problem through graph isomorphisms does not seem viable, as small perturbation of the function $f$ may create an arbitrary number of new branches and loops in the graph. 
To this end, we first put the following metric structure on a Reeb graph~$\rg_f$ to capture information about the function~$f$ . 

Specifically, 
for any two points $u, v \in \rg_f$ (not necessarily nodes), let $\pi$ be a continuous path between $u$ and $v$. The \emph{range} of this path is the interval $\range(\pi) := [\min_{x\in \pi} f(x), \max_{x\in \pi} f(x)]$, and its \emph{height} is simply the length of the range, denoted by $ \height(\pi) = \max_{x\in \pi} f(x) - \min_{x\in \pi} f(x)$. 
We define the distance 
\begin{equation}
\dreeb_f(u,v) = \min_{\pi: u \leadsto v} \height(\pi),
\label{eqn:df}
\end{equation}
where $\pi$ ranges over all paths from $u$ to $v$, denoted by $u \leadsto v$.
Equivalently, $\dreeb_f(u,v)$ is the minimum length of any interval $I$ such that $u$ and $v$ are in the same connected component of $f^{-1}(I)$.
Note that this is in fact a metric, since on Reeb graphs there is no path of constant function value between two points $u\neq v$.
We put~$f$ in the subscript to emphasize the dependency on the input function. 
Intuitively, $\dreeb_f(u,v)$ is the minimal function difference one has to overcome to move from $u$ to $v$. 

To define a distance between $\rg_f$ and $\rg_g$, we need to connect the spaces $\rg_f$ and $\rg_g$, 
which is achieved by continuous maps $\leftmap: \rg_f \rightarrow \rg_g$ and $\rightmap: \rg_g \rightarrow \rg_f$.  
Borrowing from the definition of Gromov--Hausdorff distance given in~\cite{Kalton2008Distances},
let 
\begin{align}
G(\leftmap,\rightmap)&=\big\{(x,\leftmap(x)):x\in\rg_f\}\cup\{(\rightmap(y),y):y\in\rg_g\big\} ~~\text{and} \nonumber \\
D(\leftmap,\rightmap)&=\sup_{(x,y),(\tilde x,\tilde y)\in G(\leftmap,\rightmap)}\frac12\left|\dreeb_f(x,\tilde x)-\dreeb_g(y,\tilde y)\right|, 
\label{eqn:GD}
\end{align}
where $G(\leftmap,\rightmap)$, the union of the graphs of $\phi$ and $\psi$, can be thought of as the set of correpondences between $\rg_f$ and $\rg_g$ induced by maps $\leftmap$ and $\rightmap$. 
The \emph{\GHlike{} distance} is defined as: 
\begin{align}
\DD(\rg_f, \rg_g) = \inf_%
{\leftmap,\rightmap}
\max \big\{ D(\leftmap,\rightmap), \|f-g\circ\leftmap\|_\infty, \|f\circ\rightmap-g\|_\infty \big\},  
\label{eqn:distdef}
\end{align}
where $\leftmap$ and $\rightmap$ 
range over all continuous maps between $\rg_f$ and $\rg_g$. 
The latter two terms address the fact that composition with isometries of the real line (translation, negation) does not affect the metric $\dreeb_f$ induced by a function $f$.
Note that this definition can be considered as a continuous, functional variant of the Gromov--Hausdorff distance, with the additional condition that the maps between $\rg_f$ and $\rg_g$ are required to be continuous, and taking into consideration the difference between the function values of corresponding points as well. In fact, this definition is the continuous version of the extended Gromov-Hausdorff distance introduced in Definition 2.4 of \cite{CCGMO09}. 
Furthermore, it turns out that for metric graphs, our continuous version of the extended Gromov-Hausdorff (GH) distance is a constant factor approximation of the extended GH distance induced by arbitrary maps, which we will make precise and show later in Section \ref{sec:GHrelation}. 
As an example, consider the two trees in \cref{fig:twotypes}. The distortion of distances in the two trees in (a) is large no matter how we identify correspondences between points from them. Thus the \GHlike{} distance between them is also large, making it more discriminative than the bottleneck distance between persistence diagrams. 

It is straightforward to show that the functional distortion distance is a pseudometric, and a metric on the equivalence classes of Reeb graphs up to function-preserving homeomorphisms.
Note that this definition and our results apply to any graph $G$ with a function $f$ that is strictly monotonic on the edges. This is easy to see since in that case $\rg_f=G$ and $\tilde f=f$.

\section{Properties of the Functional Distortion Distance}
\label{sec:properties}

In this section, we show that the \GHlike{} distance is both stable (upper bounded)
and discriminative (lower bounded). Note that it is somewhat meaningless to discuss the stability of a distance alone without understanding its discriminative power -- the constant function with value $0$ is a pseudo-metric too. 

\subsection{Stability}
\label{sec:traditional}
Suppose that $f$ and $g$ are defined on the same domain $\XX$. 
Furthermore, assume that the quotient maps $\mu_f$ and $\mu_g$ have continuous sections (right-inverses) $s_f$ and $s_g$, i.e., $\mu_f \circ s_f = \id_{\rg_f}$ and $\mu_g \circ s_f = \id_{\rg_f}$. 
Then we have the following stability result for the metric $\DD$ for Reeb graphs. 
\begin{theorem}
Let $f, g: \XX \to \reals$ be tame functions whose Reeb quotient maps $\mu_f: \XX \to \rg_f$ and $\mu_g: \XX\to \rg_g$ have continuous sections. Then $\DD(\rg_f, \rg_g) \le \|f - g\|_\infty$.  
\label{thm:stability}
\end{theorem}

\begin{proof}
Let $\delta=\|f - g\|_\infty$.
Choose $\leftmap=\mu_g \circ s_f,\rightmap=\mu_f \circ s_g$.
Now assume that $(x,y),(\tilde x,\tilde y)\in G(\leftmap,\rightmap)$, with $G(\leftmap, \rightmap)$ as defined in \cref{eqn:GD}.
Let $\xi=s_f(x)$, $\tilde\xi=s_f(\tilde x)$, $\upsilon=s_g(y)$, and $\tilde\upsilon=s_g(\tilde y)$.
Note that either $y=\leftmap(x)$ or $x=\rightmap(y)$, so either
\[\mu_g(\upsilon)=y=\leftmap(x)=\mu_g \circ s_f(x)=\mu_g(\xi)\]
or
\[\mu_f(\xi)=x=\rightmap(y)=\mu_f \circ s_g(y)=\mu_f(\upsilon).\]
In other words,
$\xi$ and $\upsilon$ are either in the same level set component of $f$ or of $g$, and analogously for $\tilde\xi$ and $\tilde\upsilon$.

Let $[a,b]$ be such that $x,\tilde x$ are connected in $\tilde f^{-1}[a,b]$. 
Then $\xi$ and $\tilde\xi$ are connected in \[f^{-1}[a,b]\subset g^{-1}{[a-\delta,b+\delta]},\] and hence, by the above, $\upsilon$ and $\tilde\upsilon$ are also connected in $g^{-1}{[a-\delta,b+\delta]}$. Therefore, $y$ and $\tilde y$ are connected in $\tilde g^{-1}{[a-\delta,b+\delta]}$. We conclude that $(b-a)+2\delta \geq \dreeb_g(y,\tilde y)$. Since this inequality holds for all intervals $[a,b]$ with the stated properties, we have $\dreeb_f(x,\tilde x)+2\delta \geq \dreeb_g(y,\tilde y)$. By symmetry of the above argument, we also have $\dreeb_g(y,\tilde y)+2\delta \geq \dreeb_f(x,\tilde x)$.
Moreover, by assumption, \[\max_{x \in \rg_f} | f(x) - g\circ \leftmap(x)| \le \max_{y \in \XX} |f(y) - g(y)| =\delta.\] Similarly, \[\max_{x \in \rg_g} | g(x) - f\circ \rightmap(x)| \le \max_{y \in \XX} |g(y) - f(y)| =\delta.\] Hence $\| f- g\circ\leftmap\|_\infty \le \delta$ and $\|f\circ\rightmap - g \|_\infty \le \delta$. Combining these with \cref{eqn:distdef}, we conclude that $\DD(\rg_f, \rg_g) \le \|f - g\|_\infty$.
\end{proof}

The above result is similar to the stability result obtained for the bottleneck distance between persistence diagrams~\cite{CEH09}, as well as for the $\eps$-interleaving distance between merge trees~\cite{MBW13}. 
Note that the above stated conditions (on the existence of continuous sections) are only required for the stability result. They are not necessary for \cref{thm:traditional,thm:extendedbound}. The condition on the common domain $X$ is required so that we can define the distance between input scalar fields $f$ and $g$. The condition on the existence of sections is purely technical; it holds e.g.\@ for Morse functions or for generic PL functions.

\subsection{Relation to Ordinary Persistence Diagram}

The main part of this section is devoted to discussing the discriminative power of the \GHlike{} distance for Reeb graphs. In particular, we relate this distance with the bottleneck distance between persistence diagrams. 
We have already seen in \cref{fig:twotypes} (a) that there are cases where the \GHlike{} distance is strictly larger than the bottleneck distance between persistence diagrams of according dimensions ($0$th ordinary and $1$st extended persistence diagrams). 
We next show that, up to a constant factor, the \GHlike{} distance is always at least as large as the bottleneck distance. 
We take different approaches to investigate the branching features (ordinary persistence diagram) and the \mycycle{} features (extended persistence diagram). 
For the former, we have the following main result. The proof is rather standard, and similar to the result on interleaving distance between merge trees in \cite{MBW13}. %

\begin{theorem}
$d_B(\Dg_0(\rg_f), \Dg_0(\rg_g)) \le \DD(\rg_f, \rg_g).$ ~Similarly, $d_B(\Dg_0(\rg_{-f}), \Dg_0(\rg_{-g})) \le \DD(\rg_f, \rg_g)$.
\label{thm:traditional}
\end{theorem}

\begin{proof}
Let $\leftmap: \rg_f \rightarrow \rg_g$ and $\rightmap: \rg_g \rightarrow \rg_f$  be the optimal continuous maps that achieve $\delta = \DD(\rg_f, \rg_g)$ \footnote{If the $\DD(\rg_f, \rg_g)$ is achieved only in the limit, then one can extend the argument by constructing two sequences of maps that are optimal up to an arbitrarily small additive term $\eps$ and taking the limit in the distance they induce.}. 
First, note that by \cref{eqn:distdef}, $\max_{x\in \rg_f} |f(x) - g(\leftmap(x))| \le \optd$. Hence $\leftmap:(\rg_f)_{\leq\alpha} \to (\rg_g)_{\leq\alpha+\optd}$ is well defined for any $\alpha\in\R$. 
Similarly,  $\rightmap:(\rg_g)_{\leq\beta} \to (\rg_f)_{\leq \beta+\optd}$ is well defined for any $\beta\in\R$.
Let $i$ denote the canonical inclusion maps, and for any map $\rho$, let $\rho_*$ indicate the induced homomorphism on homology. 
We now show that the following %
diagram commutes %
for any real value $\alpha$: 
\[
\begin{tikzcd}[column sep=-16pt]
H_0\left((\rg_f)_{\leq\alpha}  \arrow{rr}{i_*} \ar{dr}{\leftmap_*} \right) & & H_0\left((\rg_f)_{\leq\alpha+2\optd} \right) \\
 & H_0\left((\rg_g)_{\leq\alpha+\optd} \ar{ur}{\rightmap_*} \right) 
\end{tikzcd}
\label{eqn:diag1}
\]

To show the commutativity of the above %
diagram, 
we need to show that for any $0$-cycle $c$ in $(\rg_f)_{\leq\alpha}$, $[i(c)] = [\rightmap \circ \leftmap (c)]$, where $[c']$ is the homology class represented by a cycle $c'$. 
Assume w.l.o.g.\ that the $0$-cycle $c = x_1 + x_2$ contains only two points $x_1, x_2$ from $(\rg_f)_{\leq\alpha}$; 
the argument easily extends to the case where $c$ contains an arbitrary even number of points. 
Let $x_1' = \rightmap \circ \leftmap (x_1)$ and $x_2' = \rightmap \circ \leftmap (x_2)$. 
Since $\dreeb_f(x_1, x_1') \le \optd$, we know that there is a path (1-chain) $\pi(x_1, x_1')$ with height at most $\optd$ connecting $x_1$ and $x_1'$. In other words, $x_1$ and $x_1'$ are connected in $(\rg_f)_{\leq\alpha+\optd} \subseteq (\rg_f)_{\leq \alpha+2\optd}$. Similarly, $x_2$ and $x_2'$ are connected in $(\rg_f)_{\leq\alpha+2\optd}$. 
Hence the new $0$-cycle $c' = x_1' + x_2' = \rightmap \circ \leftmap (c)$ is homologous to $c$ in $(\rg_f)_{\leq\alpha+2\optd}$. Thus, $[i(c)] = [c'] = [\rightmap \circ \leftmap (c)]$. 

A similar argument also shows that the symmetric versions of the diagrams in \cref{eqn:diag1} (by switching the roles of $\rg_f$ and $\rg_g$) also commute at the 0th homology level. This means that the two persistence modules $\{ \HH_0( (\rg_f)_{\leq \alpha}) \}_\alpha$ and $\{\HH_0 ((\rg_g)_{\leq \beta}) \}_\beta$ are strongly $\optd$-interleaved (as introduced in \cite{CCG09}). The first half of \cref{thm:traditional} then follows from Theorem 4.8 of \cite{CCG09}. 

The same argument works for the scalar fields $-\tilde f: \rg_f \to \reals$ and $-\tilde g: \rg_g \to \reals$, which proves the second half of \cref{thm:traditional}. Recall that $\Dg_0(\rg_f)$ captures minimum and down-fork persistence pairs, while $\Dg_0(\rg_{-f})$ captures up-fork and maximum persistence pairs.
\end{proof}

\subsection{Relation to Extended Persistence Diagram}
\label{sec:extended}
Recall that the range of \mycycle{} features 
in the Reeb graph correspond to points in the 1st extended persistence diagram. 
In what follows we will show the following main theorem, which states that $\DD(\rg_f, \rg_g)$ is bounded below by the bottleneck distance between the 1st extended persistence diagrams $\eDg_1(\rg_f)$ and $\eDg_1(\rg_g)$. 
\begin{theorem}
$d_B( \eDg_1(\rg_f), \eDg_1(\rg_g) ) \le 3\DD(\rg_f, \rg_g).$
\label{thm:extendedbound}
\end{theorem}

For simplicity of exposition, we assume that $\DD(\rg_f, \rg_g)$ can be achieved by optimal continuous maps $\leftmap: \rg_f \rightarrow \rg_g$ and $\rightmap: \rg_g \rightarrow \rg_f$. The case where $\DD(\rg_f, \rg_g)$ is achieved in the limit can be handled by considering a sequence of continuous maps that are optimal up to an arbitrarily small additive term $\epsilon$. 
Let $\optd = \DD(\rg_f, \rg_g)$.

\paragraph{\Mycanonical{} bases}
Let $\ZZ_1(\rg_f)$ be the $1$-dimensional cellular cycle group of $\rg_f$ with coefficients in $\Z_2$, i.e., the subgroup of the $1$-dimensional cellular chains with zero boundary. 
Since the Reeb graph has the structure of a 1-dimensional CW complex, the 1-dimensional cellular boundary group is trivial, and so every \emph{cellular} 1-cycle 
of $\rg_f$ represents a unique homology class\footnote{Note that the same is not true for singular homology; this is the reason why we consider cellular cycles here.} in $\HH_1(\rg_f)$; that is, $\HH_1(\rg_f) \cong \ZZ_1(\rg_f)$.  

For a cellular 1-cycle $\gamma=\sum_\alpha e_\alpha$, let $\im \gamma$ denote the union of the images of the characteristic maps %
for all 1-cells (edges) $e_\alpha$. %
Let $\range(\gamma) = [ \min_{x\in \im \gamma} f(x), \max_{x \in \im \gamma} f(x)]$ denote the \emph{range} of a cycle $\gamma$, and let $\height(\gamma)$ be the length of this interval. 
A cycle is \emph{thinner} than another one if its height is strictly smaller. 
A cycle $\gamma$ is \emph{\mycanonical{}} if it cannot be written as a linear combination of thinner cycles. 
See \cref{fig:twotypes} (b), where the cycle $x_4x_8x_6x_5x_9x_4$ is \mycanonical{}, while the cycle $x_3x_4x_9x_5x_6x_3$ is not. 
Given a basis of $\ZZ_1(\rg_f)$, 
consider the sequence of the heights of the cycles contained in it, ordered in non-decreasing order. 
A basis for $\ZZ_1(\rg_f)$ is a \emph{\mycanonical{}} basis if its height sequence is 
less than or equal to that of any other basis of $\ZZ_1(\rg_f)$ in the lexicographic order. %
Obviously, each cycle in a \mycanonical{} basis is necessarily a \mycanonical{} cycle. 

From now on, we fix an arbitrary \mycanonical{} basis $\gset_f = \{ \loopone_1, \ldots, \loopone_\frank \}$ of $\ZZ_1(\rg_f)$ and $\gset_g = \{ \looptwo_1, \ldots, \looptwo_\grank \}$ of $\ZZ_1(\rg_g)$, with $\frank$ and $\grank$ being the rank of $\ZZ_1(\rg_f)$ and $\ZZ_1(\rg_g)$, respectively. 
It is known \cite{CEH09} that every cycle in a \mycanonical{} basis of $\rg_f$ is necessarily a \mycanonical{} cycle, and the ranges $[b, d]$ of cycles in $\gset_f$ (resp.~in $\gset_g$) correspond one-to-one to the points $(b, d)$ in the 1st extended persistence diagram $\eDg_1(\rg_f)$ (resp.~in $\eDg_1(\rg_g)$). 
For example, in \cref{fig:twotypes} (b), the two cycles $x_3x_4x_8x_6x_3$ and $x_4x_8x_6x_5x_9x_4$ form a \mycanonical{} basis, corresponding to points $(\tilde f(x_8), \tilde f(x_3))$ and $(\tilde f(x_9), \tilde f(x_4))$ in $\eDg_1(\rg_f)$ in (d). 

Given any cycle $\aloop$ of $\rg_f$ (resp.~of $\rg_g$), we can represent $\aloop$ uniquely as a linear combination of cycles in $\gset_f$ (resp.~$\gset_g$), which we call the \emph{\mydecomposition{} of $\aloop$}; we omit the reference to $\gset_f$ and $\gset_g$ since they will be fixed from now on.  
The \mycanonical{} cycle with the largest height from the \mydecomposition{} of $\aloop$ is called the \emph{dominating cycle of $\aloop$}, denoted by $\domcycle(\aloop)$. 
If there are multiple cycles with the same maximal height, then by convention we choose the one with smallest index in $\gset_f$ (resp.~in $\gset_g$) as the dominating cycle. 
A cycle $\aloop$ is \emph{$\alpha$-stable} if its dominating cycle has a height strictly larger than $2\alpha$.
Let $\ZZ_1^\alpha (\rg_f)$ %
denote the subgroup of $\ZZ_1(\rg_f)$ generated by cycles with height at most $2\alpha$. Equivalently, a thin basis decomposition of a cycle in $\ZZ_1^\alpha(\rg_f)$ consists only of cycles with height at most $2\alpha$. Hence, a cycle $z$ is in $\ZZ_1^\alpha(\rg_f)$ if and only if $z$ is \emph{not} $\alpha$-stable. Note that this only means that the \emph{dominating} cycle of $z$ has height at most $2\alpha$; the height of $z$ itself can be larger than~$2\alpha$.
We have the following property of the dominating cycle:
\begin{lemma}
A set of cycles $\aloop_1, \ldots \aloop_k \in \ZZ_1(\rg_f)$ with distinct dominating cycles is linearly independent. 
\label{lemma:independent}
\end{lemma}

\begin{proof}
We show that $\sum_{j=1}^a \aloop_{i_j} \neq 0$ for any subset $\{i_1, \ldots, i_a \} \subseteq \{1, 2, \ldots, k\}$. 
Specifically, consider the maximum height of dominating cycles of any cycle in $\{ \aloop_{i_j} \}_{j=1}^a$.
First, assume that there is only a unique cycle, say $\aloop_{i_a}$, whose dominating cycle $\domcycle(\aloop_{i_a})$ has this maximum height. 
It then follows that this \mycanonical{} cycle $\domcycle(\aloop_{i_a})$ is not in the  \mydecomposition{} of any other cycle $\aloop_{i_j}$, $j \neq a$. 
Since the  \mydecomposition{} of the cycle $\sum_{j=1}^a \aloop_{i_j}$ is simply the sum (modulo 2) of  \mydecomposition{} of each $\aloop_{i_j}$, $\domcycle(\aloop_{i_a})$ must exist in the  \mydecomposition{} of the cycle $\sum_{j=1}^a \aloop_{i_j}$; in fact, $\domcycle(\sum_{j=1}^a \aloop_{i_j}) = \domcycle(\aloop_{i_a})$. Therefore $\sum_{j=1}^a \aloop_{i_j} \neq 0$. 

If there are multiple cycles whose dominating cycle has the maximal height, then we consider the one whose dominating cycle has the smallest index among all of them. The same argument as above shows that this cycle will present in the  \mydecomposition{} of the cycle $\sum_{j=1}^a \aloop_{i_j}$, implying that $\sum_{j=1}^a \aloop_{i_j} \neq 0$. 
\end{proof}

\paragraph{$\alpha$-matching}
The main use of \mycanonical{} cycles is that we will prove \cref{thm:extendedbound} by showing the existence of an \emph{$\alpha$-matching} between $\gset_f$ and $\gset_g$. 
Specifically, two \mycanonical{} cycles $\aloop_1$ and $\aloop_2$ are 
\emph{$\alpha$-close} if their ranges $[a, b]$ and $[c, d]$ are within Hausdorff distance $\alpha$, i.e., $|c -a| \le \alpha$ and $|d - b| \le \alpha$. (Note that two $\alpha$-close cycles can differ in height by at most $2\alpha$.) 
An \emph{$\alpha$-matching} for $\gset_f$ and $\gset_g$  is a set of pairs
$\mathcal{M} \subset \gset_f \times \gset_g$ such that:
\begin{enumerate}[(I)]
\item For each pair $( \loopone, \looptwo ) \in \mathcal{M}$, the cycles $\loopone$ and $\looptwo$ are $\alpha$-close; and 
\item Every $\alpha$-stable cycle in $\gset_f$ and $\gset_g$ (i.e., with height larger than $2\alpha$) appears in \emph{exactly one} pair of $\mathcal{M}$; every other cycle appears in at most one pair.%
\end{enumerate}
Since each point $(b,d)$ in the extended persistence diagram corresponds to the range $[b,d]$ of a unique cycle in a given \mycanonical{} basis, $d_B(\eDg_1(\rg_f), \eDg_1(\rg_g)) \le \alpha$ if and only if there is an $\alpha$-matching for $\gset_f$ and $\gset_g$. 
Our goal now is to prove that there exists a $3\delta$-matching for $\gset_f$ and $\gset_g$, which will then imply \cref{thm:extendedbound}.

\paragraph{Properties of $\leftmap$ and $\rightmap$}
Recall that $\leftmap: \rg_f \to \rg_g$ and $\rightmap: \rg_g \to \rg_f$ are the optimal continuous maps that achieve $\delta=\DD(\rg_f, \rg_g)$. 
We now investigate the effect of these maps on \mycanonical{} cycles. 
Note that $\leftmap$ and $\rightmap$ induce maps $\ZZ_1(\rg_f)\to\ZZ_1(\rg_g)$ and $\ZZ_1(\rg_g)\to\ZZ_1(\rg_f)$. To simplify notation, we denote these maps by $\leftmap$~and~$\rightmap$ as well.
\Cref{lem:cycledistortion} below states that $\rightmap$ is ``close'' to being an inverse of $\leftmap$. \Cref{lem:rangebound,lem:shortenheight} relate the range of $\leftmap(\gamma)$ with the range of $\gamma$. 
\begin{lemma}
Given any cycle $\aloop \in \ZZ_1(\rg_f)$, we have $\rightmap \circ \leftmap (\aloop) \in \aloop + \smallH(\rg_f)$. That is, $\rightmap\circ \leftmap (\aloop) = \aloop + \aloop'$, where $\aloop'$ is not $2\delta$-stable. A symmetric statement holds for any cycle in $\rg_g$. 
\label{lem:cycledistortion}
\end{lemma}

\begin{proof}
The input Reeb graph $\rg_f$ is a finite graph, and there are only a finite number of (cellular) cycles. Hence there are only a finite number of height values (reals) that these cycles can have.
Let $\alpha$ denote the lowest height of any cycle whose height is strictly larger than $4 \delta$.
We set $\rho$ to be any positive constant between $0$ and $\alpha - 4\delta$; that is, $4\delta +\rho < \alpha$. 
Note that if a cycle $\gamma$ satisfies $\height(\gamma) \le 4\delta + \rho < \alpha$, then it is necessary that $\height(\gamma) \le 4\delta$.

Assume that $\im\aloop$ has only a single connected component -- the case with multiple components can be handled in a component-wise manner. This means that $\aloop$ is generated by a loop $\ell$, i.e., a closed curve on $\rg_f$, in the sense that we can consider $\ell$ as a singular cycle and then use the isomorphism of singular homology and cellular cycles. 
Without loss of generality, we may assume that $\ell$ is embedded; otherwise, $\ell$ can be split into embedded subloops, which again can be treaded separately.
Now for the given loop $\ell$, let $\tilde \ell$ denote $\rightmap \circ \leftmap (\ell)$, and for any point $x$ on $\ell$, let $\tilde x \in \tilde \ell$ denote $\tilde x := \rightmap \circ \leftmap (x)$.
Since both $(x, \leftmap(x))$ and $(\tilde{x}, \leftmap(x))$ are in $G(\leftmap, \rightmap)$, by \cref{eqn:distdef}, there is an embedded path $\pi(x, \tilde{x})$ connecting $x$ to $\tilde x$ with $\range(\pi(x, \tilde x))=[a,b]$, where $b-a \le 2\delta$.

Let $\ell[x, x']$ and $-\ell[x, x']$ denote the orientation-preserving and orientation-reversing subcurves of $\ell$ from $x$ to $x'$, respectively.
Start with an arbitrary point $x= x_0$ on $\ell$, and consider $\tilde x_0$ on $\tilde \ell$.
Since $\rightmap \circ \leftmap$ is a continuous map, as we move $x$ along $\ell$ continuously, $\tilde x$ moves continuously.
In step $i$, as we move along $\ell$ (starting from $x_i$), we set $x_{i+1}$ to be the first point such that the height of the loop $\ell_i = \ell[x_i, x_{i+1}] \circ \pi(x_{i+1}, \tilde x_{i+1}) \circ - \tilde \ell [\tilde x_{i+1}, \tilde x_i] \circ \pi(\tilde x_i, x_i)$ is $4\delta + \rho$.
If no such~$\ell_i$ exists before $\tilde{x}$ moves back to $\tilde{x_0}$, then we set $x_{i+1}$ to be $x_0$, and the process terminates.
Since both $\pi(x_{i+1}, \tilde x_{i+1})$ and $\pi(\tilde x_i, x_i)$ are paths of height at most $2\delta$, the sum of the heights of $\ell[x_i, x_{i+1}]$ and $-\tilde \ell [\tilde x_i, \tilde x_{i+1}]$ must be at least $\rho$. Hence, for a fixed value $\rho$, this process terminates in a finite number of steps.
Now let~$c_i$ denote the cellular cycle homologous to the loop $\ell_i$ in the above construction.
By construction, we have that $\aloop = \tilde\aloop + \sum_{i} c_i$, where $\tilde\aloop = \rightmap \circ \leftmap (\aloop)$.
Since each $c_i$ satisfies $\height(c_i) \le \height(\ell_i) \le 4\delta + \rho < \alpha$, as discussed earlier it then follows that $\height(c_i) \le 4\delta$. Hence $c_i \in \smallH (\rg_f)$ for each $i$ and $\aloop' = \sum_{i} c_i \in \smallH(\rg_f)$, implying that $\rightmap \circ \leftmap(\aloop) = \tilde \aloop = \aloop + \aloop' \in \aloop + \smallH (\rg_f)$. 
\end{proof}

\begin{lemma}
Given any \mycanonical{} cycle $\loopone \in \ZZ_1(\rg_f)$ with range $[b, d]$, we have that the range of any cycle in the \mydecomposition{} of $\leftmap(\gamma)$ must be contained in the interval $[b-\delta, d+\delta]$. 
\label{lem:rangebound}
\end{lemma} 

\begin{proof}
First, by \cref{eqn:distdef}, we have $\max_{x\in \rg_f} |f(x) - g(\leftmap(x))| \le \delta$.
Hence $\range(\leftmap(\gamma)) \subseteq [b-\delta, d+\delta]$. 
Now let $b'$ be the smallest left endpoint of the range of any cycle in the  \mydecomposition{} of $\leftmap(\gamma)$.
We will prove that $b' \ge b - \delta$. 

Suppose this is not case and $b' < b - \delta$. 
Now let $\looptwo_{i_1}, \ldots, \looptwo_{i_a} \in \gset_g$, $a \ge 1$, denote all those cycles in the \mydecomposition{} of $\leftmap(\gamma)$ whose ranges have $b'$ as the left endpoint. 
Set $\rho = \looptwo_{i_1} + \cdots \looptwo_{i_a}$. 
Assume $\looptwo_{i_a}$ has the largest height among these \mycanonical{} cycles. 
Note that $\range(\looptwo_{i_j}) \subseteq \range(\looptwo_{i_a})$ for any $j < a$, as all these ranges share the same left endpoint $b'$. 
Clearly $\range(\rho) \subseteq \range(\looptwo_{i_a})$. 
On the other hand, by definition of $b'$ and $\looptwo_{i_j}$, all other cycles in the \mydecomposition{} of $\leftmap(\gamma)$ have a range whose left endpoint is strictly greater than~$b'$. Let $\rho'$ be the sum of these other cycles; we have that $\leftmap(\gamma) = \rho + \rho'$. Since the left endpoint of $\range(\rho')$ is strictly bigger than $b'$, the left endpoint of $\range(\rho)$ also has to be strictly bigger than $b'$, as otherwise the left endpoint of $\range(\rho + \rho')$ would be $b'$, which contradicts to the fact that the left endpoint of $\range(\leftmap(\gamma))$ is at least $b - \delta > b'$. 
In other words, it is necessary that $\range(\rho)$ is a proper subset of $\range(\looptwo_{i_a})$; i.e., $\range(\rho) \subset \range(\looptwo_{i_a})$. 
This implies that $\rho$ has strictly smaller height than $\looptwo_{i_a}$. 
This however contradicts that $\gset_g$ is a thin basis, because we can replace $\looptwo_{i_a}$ in $\gset_g$ with $\rho$ and obtain a basis element with smaller height (the resulting set of cycles remain independent). 
Therefore it is not possible that $b' < b-\delta$. 

A symmetric argument shows that the largest right endpoint of the range of any cycle in the  \mydecomposition{} of $\leftmap(\gamma)$ is at most $d + \delta$. 
Hence the range of any cycle in the  \mydecomposition{} of $\leftmap(\gamma)$ is a subset of $[b-\delta, d+\delta]$. 
\end{proof}

\begin{lemma}
For any $2\delta$-stable cycle $\aloop \in \ZZ_1(\rg_f)$, we have \[\height(\domcycle(\leftmap(\aloop))) \geq \height(\domcycle(\aloop)) - 2\delta.\] A symmetric statement holds for any cycle of $\rg_g$. 
\label{lem:shortenheight}
\end{lemma}

\begin{proof}
Note that in this lemma, $\aloop$ is not necessarily a thin loop. 
Let $\loopone_s = \domcycle(\aloop)$; since $\aloop$ is $2\delta$-stable, we have $\height(\loopone_s) > 4\delta$. 
First, we claim that $\domcycle(\rightmap \circ \leftmap (\aloop)) = \loopone_{s}$. 
This is because by \cref{lem:cycledistortion}, $\rightmap\circ \leftmap(\aloop) \in \aloop + \smallH(\rg_f)$. Since $\height(\loopone_{s}) > 4 \delta$, $\loopone_s$ still belongs to the  \mydecomposition{} of $\rightmap\circ \leftmap(\aloop)$ and still has the largest height. 

Now set $\tilde \looptwo = \leftmap(\aloop)$ with $\looptwo_{i_1} + \cdots + \looptwo_{i_a}$ being its  \mydecomposition{}. 
Observe that for any cycle $\aloop'$ in $\rg_f$, we have that $\height(\leftmap(\aloop')) \le \height(\aloop') + 2 \delta$, which follows from the fact that for any $x\in \rg_f$, $| f(x) - g(\leftmap(x))| \le \delta$; recall \cref{eqn:distdef}. A symmetric statement holds for a loop from $\rg_g$. %
We thus have: 
\begin{align}
\height(\domcycle(\rightmap(\tilde \looptwo))) & = \height(\domcycle( \sum_{j=1}^a \rightmap (\looptwo_{i_j}))) 
\le \max_{j=1}^a \height( \rightmap (\looptwo_{i_j})) \nonumber \\
& \le \max_{j=1}^a [\height(\looptwo_{i_j}) + 2 \delta] = \max_{j=1}^a \height(\looptwo_{i_j}) + 2\delta \nonumber\\
&= \height(\domcycle(\tilde \looptwo)) + 2\delta. 
\label{eqn:one}
\end{align}
Since we have shown earlier that $\domcycle(\rightmap \circ \leftmap(\loopone)) = \loopone_s$, it follows that $\domcycle(\rightmap (\tilde \looptwo)) = \domcycle(\rightmap \circ \leftmap(\loopone)) = \loopone_{s} = \domcycle(\aloop)$. 
Combining this with \cref{eqn:one}, we have 
\begin{align*}
\height(\domcycle(\leftmap(\aloop)) &= \height(\domcycle (\tilde \looptwo)) \ge \height(\domcycle(\rightmap(\tilde \looptwo)))- 2\delta \\
&= \height(\domcycle(\aloop)) - 2\delta,
\end{align*}
which proves the lemma.
\end{proof}

In fact, if $\gamma$ is a \mycanonical{} cycle, \cref{lem:shortenheight} can be strengthened to show that $\domcycle(\leftmap(\gamma))$ is $\delta$-close to $\gamma$ . 
This already provides some mapping of base cycles from $\gset_f$ to cycles from $\gset_g$ such that each pair of corresponding cycles are $\delta$-close. 
However, the main challenge is to show that there exists a \emph{one-to-one} correspondence for all $3\delta$-stable cycles  (recall the definition of a $3\delta$-matching of $\gset_f$ and $\gset_g$).
For this, we need to take a slight detour to relate cycles in $\gset_f$ with those in $\gset_g$ in a stronger sense: 
\begin{prop}
For any \mycanonical{} cycle $\gamma_k \in \gset_f$, one can compute a (not necessarily \mycanonical{}) cycle $\widehat{\gamma}_k$ such that $\loopone_k=\domcycle(\widehat\loopone_k)$ and
\[\leftmap(\widehat\loopone_k) \in \sum_{j=1}^r \looptwo_{k_j} + \smallH (\rg_g),\]
where each $\looptwo_{k_j} \in \gset_g$ is $3\delta$-close to $\gamma_k$, for any $j \in [1, r]$. 
\label{prop:either_large_or_small}
\end{prop}
\begin{proof}
Assume that the \mydecomposition{} of $\leftmap(\gamma_k)$ has the form: 
\[
\leftmap(\gamma_k) \in \sum_{j=1}^r \looptwo_{k_j} + \sum_{j=1}^s \tilde{\looptwo}_j + \smallH(\rg_g), 
\]
where the first term contains the $2\delta$-stable \mycanonical{} cycles whose range is $3\delta$-Hausdorff close to $\range(\gamma_k)$, the last term contains the \mycanonical{} cycles that are not $2\delta$-stable, 
and the middle term contains those \mycanonical{} cycles $\tilde{\looptwo}_j$ that are neither $3\delta$-close to $\gamma_k$ nor small. 
We wish to get rid of the middle term $\tilde{\looptwo} = \sum_{j=1}^s \tilde{\looptwo}_j$. 
Set $\gamma' =\rightmap(\tilde\looptwo)$. By \cref{lem:cycledistortion}, we have that $\leftmap(\gamma') = \leftmap \circ \rightmap(\tilde{\looptwo}) \in \tilde{\looptwo} + \smallH(\rg_g)$. 
Set $\widehat{\gamma}_k := \gamma_k + \gamma'$. It is then easy to verify that 
$\leftmap(\widehat{\gamma}_k) = \leftmap(\gamma_k) + \leftmap(\gamma') \in \sum_{j=1}^r \looptwo_{k_j} + \smallH (\rg_g)$, as claimed. 

It remains to show that $\domcycle(\widehat{\gamma}_k) = \gamma_k$. 
Let $[b,d]= \range(\gamma_k)$. 
By \cref{lem:rangebound}, $\range(\tilde \looptwo_j) \subseteq [b-\delta, d+\delta]$, for any $j \in [1, s]$. 
Since each $\tilde \looptwo_j$ is not $3\delta$-close to $\gamma_k$, 
it is then necessary that, for any $j \in [1, s]$, either $\range(\tilde \looptwo_j) \subseteq [b-\delta, d-3\delta)$ or $\range(\tilde \looptwo_j) \subseteq (b+3\delta, d+\delta]$. 
W.l.o.g.\ assume that $\range(\tilde \looptwo_j) \subseteq [b-\delta, d-3\delta)$. 
Apply Lemma~\ref{lem:rangebound} to the cycle $\tilde \looptwo_j$. We have that the range of any cycle in the \mydecomposition{} of $\rightmap(\tilde \looptwo_j)$ is contained in $[b-2\delta, d-2\delta)$, thus its height strictly smaller than $d - b$. 
Hence all cycles in the  \mydecomposition{} of $\gamma' = \rightmap(\tilde\looptwo)$ have a height strictly smaller than $\height(\gamma_k)=d-b$. 
This means that $\gamma_k$ has the largest height among the cycles of the  \mydecomposition{} of $\widehat{\gamma}_k = \gamma_k + \gamma'$, and it is the only cycles with this largest height. %
It then follows that $\gamma_k = \domcycle(\widehat \gamma_k)$. 
\end{proof}

\begin{cor}
Let $\widehat \gset_f$ denote the set of cycles $\{ \widehat \gamma_k \}_{k=1}^\frank$, where each $\widehat \gamma_k$ is as specified in \cref{prop:either_large_or_small}.
$\widehat \gset_f$~forms a basis for $\ZZ_1(\rg_f)$. 
\label{cor:changebasis}
\end{cor}
\begin{proof}
Since the dominating cycles for cycles in $\widehat \gset_f$ are all distinct, it follows from Lemma~\ref{lemma:independent} that all cycles in $\widehat \gset_f$ are linearly independent. Hence $\widehat \gset_f$ also forms a (not necessarily \mycanonical{}) basis for $\ZZ_1(\rg_f)$. 
\end{proof} 

Let $\Fmatrix$ denote the matrix of the mapping from the base cycles in $\widehat\gset_f$ (columns, domain) to those in $\gset_g$ (rows, range) as induced by $\leftmap$, i.e., the $i$th column of $\Fmatrix$ specifies the representation of $\leftmap(\widehat\loopone_i)$ using basis elements from $\gset_g$, with $\Fmatrix_{ij} = 1$ if $\looptwo_j$ is in the  \mydecomposition{} of $\leftmap(\widehat\gamma_i)$. 
Let $\widetilde \Fmatrix$ be the submatrix of $\Fmatrix$ with columns corresponding to basis elements $\widehat\loopone_i$ that are $3\delta$-stable, and rows corresponding to basis elements $\looptwo_j$ that are $2\delta$-stable.
See \cref{fig:augmenting} (a). 
By \cref{prop:either_large_or_small}, $\widetilde \Fmatrix_{ij} = 1$ implies that the basis element $\looptwo_j \in \gset_g$ is $3\delta$-close to the basis element $\loopone_i \in \gset_f$. 
Recall that our goal is to show that there is a  $3\delta$-matching for $\gset_f$ and $\gset_g$. 
Intuitively, non-zero entries in $\widetilde \Fmatrix$ will provide potential matchings for basis elements in $\gset_f$ to establish a $3\delta$-matching between $\gset_f$ and $\gset_g$ that we need. 
\begin{lemma}
The columns of $\widetilde \Fmatrix$ are linearly independent.
\label{lem:submatrixindependent}
\end{lemma}
\begin{proof}
Consider an arbitrary subset of indices $i_1, \ldots, i_s$ whose corresponding columns are in $\widetilde \Fmatrix$ (i.e, each $\widehat \loopone_{i_a}$ is $3\delta$-stable), and let $\widehat\gamma=\widehat\loopone_{i_1}+\dots+\widehat\loopone_{i_s}$. 
We will show that $\leftmap(\widehat\gamma)$ is $2\delta$-stable; that is, $\domcycle(\leftmap(\widehat\gamma))$ has height at least~$4\delta$. 
This means that the linear combination of the corresponding columns in $\widetilde \Fmatrix$ contains a non-zero element. 
Since this holds for any subset of columns from $\widetilde \Fmatrix$, we have that the columns in $\widetilde \Fmatrix$ are linearly independent. 

It remains to show that $\leftmap(\widehat \gamma)$ is $2\delta$-stable.
Recall that for any index $a$, $\loopone_a$ is the dominating cycle of $\widehat \loopone_a$. 
Assume w.l.o.g.\ that $\loopone_{i_1}$ has the largest height among all $\loopone_{i_j}$, for $j \in [1, s]$. (If there are multiple cycles from $\{ \loopone_{i_j} \}_{j=1}^s$ having this largest height, let $\loopone_{i_1}$ be the one with smallest index.) 
From the end of the proof of \cref{prop:either_large_or_small}, 
we note that for any index $a$, $\loopone_a$ is the \emph{only} cycle with maximum height among all cycles in the  \mydecomposition{} of $\widehat \loopone_a$. In other words, all other cycles in the  \mydecomposition{} of $\widehat \loopone_a$ have strictly smaller height than $\loopone_a$. 
Putting these together, it follows that $\loopone_{i_1}$ must exist in the  \mydecomposition{} of $\widehat \gamma$ w.r.t.\ the original thinnest basis $\gset_f$ and in fact, $\loopone_{i_1} = \domcycle(\widehat \gamma)$. 
Since $\loopone_{i_1}$ is $3\delta$-stable (thus its height at least $6\delta$), it then follows from \cref{lem:shortenheight} that 
$\height(\domcycle(\leftmap(\widehat\gamma))) \ge \height(\domcycle(\widehat \gamma)) - 2\delta > 4\delta. $ So $\leftmap(\widehat \gamma)$ is $2\delta$-stable.  
\end{proof}

\begin{cor}
We can identify a unique row index $i$ for each column index $j$ in $\widetilde \Fmatrix$ such that $\widetilde \Fmatrix_{ij}= 1$. 
\label{cor:bipartitematching}
\end{cor}
\begin{proof}
View the matrix $\widetilde \Fmatrix$ as the adjacency matrix for the following bipartite graph $G = (P \cup Q, E)$, where $P$ are the columns of $\widetilde \Fmatrix$, $Q$ are the rows, and there is an edge between $p \in P$ and $q \in Q$ iff the corresponding entry in the matrix $\widetilde \Fmatrix$ is $1$. 
We now claim that there is a \emph{$P$-saturated matching} for $G$: that is, there is a matching of $G$ such that every node in $P$ is matched exactly once, and each node in $Q$ is matched at most once. 
Note that this immediately implies the claim. 
Specifically, for any subset of nodes $P' \subseteq P$, let $Q' \subseteq Q$ be the union of neighbors of nodes from $P'$. 
In other words, $Q'$ is the set of rows with at least one non-zero entry in the columns $P'$.
If $|Q'| < |P'|$, then these columns of $\widetilde \Fmatrix$ will be linearly dependent, which violates \cref{lem:submatrixindependent}. Hence we have $|Q'| \ge |P'|$.
Now by Hall's Theorem (see, e.g., Page 35 of \cite{LW92}), a $P$-satuated  matching exists for $G$. 
\end{proof}

\paragraph{Proof of \cref{thm:extendedbound}} 
Recall that by \cref{prop:either_large_or_small}, $\widetilde \Fmatrix_{ij} = 1$ implies that the cycles $\loopone_j$ and $\looptwo_i$ are $3\delta$-close. 
It follows from \cref{cor:bipartitematching} that there is an injective map ${F}$ from the set of $3\delta$-stable cycles in $\gset_f$ to the cycles in $\gset_g$ such that each pair of corresponding cycles are $3\delta$-close. 
By a symmetric argument (switching the role of $\rg_f$ and $\rg_g$), 
there is also an injective map ${G}$ from the $3\delta$-stable cycles in $\gset_g$ to cycles in $\gset_f$ where each corresponding pair of cycles are $3\delta$-close. 
However, ${F}$ and $ G$ may not be consistent and do not directly give rise to a $3\delta$-matching of $\gset_f$ and $\gset_g$ yet. 
In what follows, we will modify $ F$ to obtain another \emph{injective} map $\widehat { F}$ such that (i) any $3\delta$-stable cycle in $\gset_f$ is mapped by $\widehat { F}$ to a cycle in $\gset_g$ that is $3\delta$-close, and (ii) all $3\delta$-stable cycles in $\gset_g$ are contained in $\im \widehat { F}$, the image of $\widehat { F}$. 
Note that $\widehat { F}$ provides exactly the set of correspondences necessary in a $3\delta$-matching between $\gset_f$ and $\gset_g$. 
In particular, the injectivity of $\widehat {F}$ and these conditions guarantee that the condition (II) of a $3\delta$-matching. 
As discussed at the beginning of this section, this then means that $d_B(\eDg_1(\rg_f), \eDg_1(\rg_g)) \le 3\delta = 3\DD (\rg_f, \rg_g)$, proving \cref{thm:extendedbound}. 

\begin{figure}[t]
\begin{center}
\begin{tabular}{ccc}
\includegraphics[height = 4cm]{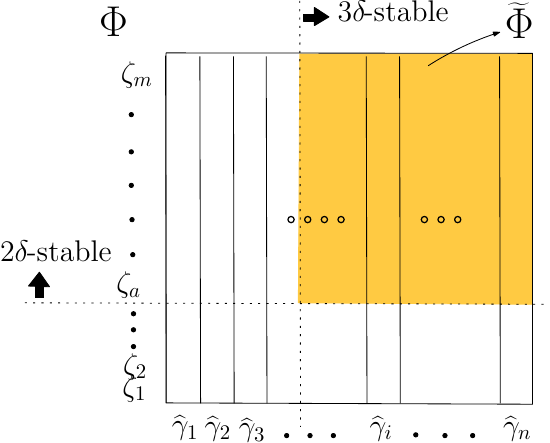} & \hspace*{0.2in} & 
\includegraphics[height=4cm]{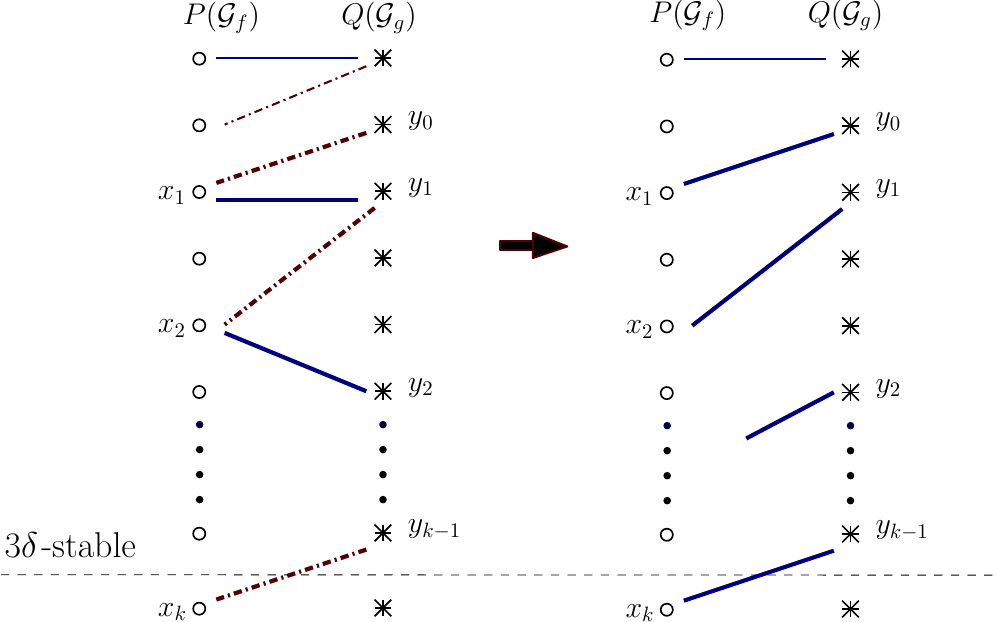} \\
(a) & & (b) \end{tabular}
\end{center}
\vspace*{-0.2in}\caption{{\small
(a): The $i$th column in the matrix $\Fmatrix$ specifies the representation of $\leftmap(\widehat \loopone_i)$ using the basis elements in $\gset_g$. The shaded submatrix represents $\widetilde \Fmatrix$. 
(b): A bipartite graph view of the augmenting process. Left: Thick path alternates between an $F$-induced (solid) and an $G$-induced (dash-dotted) edge. Right: Thick solid edges are induced by the modified injective map $\widetilde F$ (which are used to be $G$-induced edges in the left figure). 
}
\label{fig:augmenting}
}
\end{figure}

It remains to show how to construct $\widehat { F}$ satisfying conditions (i) and (ii) above. Conditions (i) already holds for $ F$, so the main task is to establish condition (ii) while maintaining (i). 
Start with $\widehat { F} =  F$. 
Let $y_0 \not \in \im \widehat{ F}$ be any $3\delta$-stable cycle in $\gset_g$ that is not yet in $\im \widehat { F}$. 
Let $x_1 =  G(y_0) \in \gset_f$. 
We continue with $y_i = \widehat { F}(x_i)$ if $x_i$ is $3\delta$-stable. 
Next, if $y_i$ is $3\delta$-stable, then set $x_{i+1} =  G(y_i)$. 
We repeat this process, until we reach $x_k$ or $y_k$ which is not $3\delta$-stable any more. 
At this time, we modify $\widehat { F}(x_i)$ to be $y_{i-1}$ for each $j \in [1, k]$ (originally, $\widehat { F}(x_i) = y_i$). 
Throughout this process, all $y_i$ 
other than $y_0$ are already in $\im \widehat { F}$. 
After the modification of $\widehat { F}$, we have $y_0 \in \im \widehat { F}$, while all other $y_i$ remain in $\im \widehat { F}$. The only exception is when the above process terminates by reaching some $y_k$ which is not $3\delta$-stable (the termination condition), in which case~$y_k$ will not be in $\im \widehat { F}$ after the modification of $\widehat { F}$.
However, the number of $3\delta$-stable cycles contained in $\im \widehat { F}$ increases by one (i.e., $y_0$) by the above process. It is easy to verify that since $F$ is injective, $\widehat F$ remains injective. 
Furthermore, $x_i$ and $y_{i-1} = \widehat { F}(x_i)$ are still $3\delta$-close, since by construction $x_i = { G}(y_{i-1})$. 

An alternative way to view this is to consider the specific bipartite graph $G' = (P \cup Q, E')$, where nodes in $P$ and $Q$ correspond to basis cycles in $\gset_f$ and $\gset_g$, respectively, and edges $E'$ are those corresponding to a cycle and its image under either the map $\widehat { F}$ or $ G$. 
The sequence $y_0, x_1, y_1, \ldots$ specifies a path with edges alternating between the $\widehat{ F}$-induced and the $ G$-induced matchings. 
The modified assignment of $\widehat{ F}(x_j)$ changes a $ G$-induced matching to an $\widehat{ F}$-induced matching along this path, much similar to the use of augmenting paths to obtain maximum bipartite matching. 
See \cref{fig:augmenting} (b) for an illustration. 

We repeat the above path augmentation process for any remaining $3\delta$-stable cycle in $\gset_g \setminus \im \widehat{ F}$. 
This process will terminate because after each augmentation process, the number of $3\delta$-stable cycles contained in $\im \widehat{ F}$ strictly increases. 
In the end, we obtain an injective map $\widehat{ F}$ from the set of $3\delta$-stable cycles of $\gset_f$ to cycles in $\gset_g$ such that $\im \widehat{ F}$ contains all $3\delta$-stable cycles of $\gset_g$. 
Hence, $\widehat{ F}$ induces a $3\delta$-matching from $\eDg_1(\rg_f)$ to $\eDg_1(\rg_g)$, finishing the proof of \cref{thm:extendedbound}. 
%

%

%


\section{Relation to Gromov--Hausdorff Distance}
\label{sec:GHrelation}

As mentioned earlier in \cref{sec:metric}, the functional distortion distance can be considered as a variant of the Gromov-Hausdorff distance (between metric spaces), restricted to continuous correspondences and taking function values into account. We now discuss this relation in more detail. 

We can view the Reeb graphs $\rg_f$ and $\rg_g$ as metric spaces, equipped with metrics $d_f$ and $d_g$, respectively. 
A natural distance for metric spaces is the \emph{Gromov--Hausdorff} distance, which, using the notation of \cref{eqn:GD}, is defined as
\begin{align}
d_{GH}(\rg_f, \rg_g) = \inf_%
{\leftmap,\rightmap}
\left( D(\leftmap,\rightmap) \right),  
\label{eqn:GHdistdef}
\end{align}
where $\leftmap: R_f\to R_g$ and $\rightmap: R_g \to R_f$ range over all maps between $R_f$ and $R_g$. Here the maps $\leftmap, \rightmap$ are not required to be continuous, which is different from our definition of the \GHlike{} distance. 

Note that translation $f+c$ and negation $-f$ do not change the metric structures of the Reeb graph $R_f$. 
To account for the difference in function values, we define the \emph{\augGH{} distance} 
 between $\onerg_f$ and $\onerg_g$, which measures not only the metric distortion but also the difference in function values between corresponding points:
 \begin{align}
d_{fGH}(\rg_f, \rg_g) := \inf_%
{\leftmap,\rightmap}
\max \left( D(\leftmap,\rightmap), \|f-g\circ\psi\|_\infty, \|f\circ\phi-g\|_\infty \right),  
\label{eqn:fGHdistdef}
\end{align}
where $\leftmap$ and $\rightmap$ 
range over all maps between $\rg_f$ and $\rg_g$. 

It turns out that we have the following relations, which imply that the \GHlike{} distance roughly measures the minimum distortion in both function values (between $f$ and $g$) and in their induced metrics (between $\dreeb_f$ to $\dreeb_g$).  

\begin{theorem}
$d_{fGH} (\rg_f, \rg_g) \le \DD(\rg_f, \rg_g) \le 3 d_{fGH} (\rg_f, \rg_g). $
\label{thm:GHrelations}
\end{theorem}
We note that a similar result also holds for the Gromov--Hausdorff distance, without the terms controlling the function values.
Specifically, for metrics $d_f$ and $d_g$ the standard Gromov--Hausdorff distance as defined in \cref{eqn:GHdistdef} is equivalent to its continuous variant up to a constant factor, restricting $\leftmap$ and $\rightmap$ to continuous maps. 
This relation does not hold in general.

\begin{proof} 

The left inequality $d_{fGH} (\rg_f, \rg_g) \le \DD(\rg_f, \rg_g)$ is immediate from the definitions.
We now prove the right inequality $\DD(\rg_f, \rg_g) \le 3 d_{fGH} (\rg_f, \rg_g)$. 
Fix an arbitrary positive real value $\eps$. Let $\CC$ denote an \emph{$\eps$-optimal correspondence}, i.e., the maximum of the three terms in the right hand side of \cref{eqn:fGHdistdef} is less than or equal to $d_{fGH}(\rg_f, \rg_g) + \eps$. 
Set $\beta = d_{fGH} (\rg_f, \rg_g)$. Our final goal is to show that $\DD(\rg_f, \rg_g) \le 3\beta$. 
We do this by constructing \emph{continuous} maps $\leftmap^\eps: \rg_f \to \rg_g$ and $\rightmap^\eps: \rg_g \to \rg_f$, based on the $\eps$-optimal pair of maps $(\leftmap,\rightmap)$ (which are not necessarily continuous), so that each of the terms in \cref{eqn:distdef} can be bounded by $3\beta+O(\eps)$. 
%

We now show how to construct a certain \emph{continuous} map $\leftmap^\eps: \rg_f \to \rg_g$ from the map $\leftmap: \rg_f \to \rg_g$. 
To do so, we will first construct an \emph{$\eps$-subdivition of $\rg_f$} as follows: 
We subdivide all arcs in $\rg_f$ to obtain a set of nodes $V_\eps = \{v_1, \ldots, v_N\}$ such that $f$ is monotonic on each resulting arc, and the height of an arc $v_i v_j$ (which is $|f(v_i) - f(v_j)|$ since $f$ is monotonic on $v_i v_j$) is at most $\eps$. 
We set $\leftmap^\eps(v_i)=\leftmap(v_i)$. 


Next, we extend this map defined on the nodes in $V_\eps$ to a continuous map defined on the entire graph $\rg_f$. 
In particular, consider an arc $v_i v_j$ and assume w.l.o.g.\ that $f(v_i) \le f(v_j)$. 
Consider $\tilde v_i = \leftmap^\eps(v_i)$ and $\tilde v_j = \leftmap^\eps(v_j)$. 
Since $(\leftmap,\rightmap)$ is $\eps$-optimal, we know that 
\[\frac12|d_f(v_i, v_j) - d_g(\tilde v_i, \tilde v_j) | \le \beta + \eps,\] 
thus 
\[d_g(\tilde v_i, \tilde v_j) \le d_f(v_i, v_j) + 2 (\beta + \eps) \le 2 \beta + 3\eps.\] 
This means that there is an embedded path $\pi(\tilde v_i, \tilde v_j)$ in $\rg_g$ connecting $\tilde v_i$ to $\tilde v_j$ whose height is at most $2 \beta + 3 \eps$. 
We now extend $\leftmap^\eps$ to an arbitrary homeomorphism from the arc $v_i v_j$ of $\rg_f$ to this path $\pi(\tilde v_i, \tilde v_j)$ with  $\leftmap^\eps(v_i) = \tilde v_i$ and $\leftmap^\eps(v_j) = \tilde v_j$. 
Assembling all these pieces of $\leftmap^\eps$ on each arc of $\rg_f$ yields the continuous map $\leftmap^\eps: \rg_f \to \rg_g$. 

Given any point $x \in \rg_f$, assume that $x$ lies on the arc $v_i v_j$. 
Then $\tilde x := \leftmap^\eps (x)$ is mapped to some point in $\pi(\tilde v_i, \tilde v_j)$. 
Since  $(\leftmap,\rightmap)$ is $\eps$-optimal, by definition in \cref{eqn:fGHdistdef}, 
\[g(\tilde v_i) \in [f(v_i) - \beta - \eps, f(v_i) + \beta + \eps ]
\quad\text{and}\quad
g(\tilde v_j) \in [f(v_j) - \beta - \eps, f(v_j) + \beta + \eps]. \]
Since the path $\pi(\tilde v_i, \tilde v_j)$ has height at most $2\beta + 3\eps$, we then have 
\[\range(\pi(\tilde v_i, \tilde v_j)) \in [f(v_i) - 3\beta - 4\eps, f(v_j) + 3\beta + 4\eps].\] 
Since $x \in v_i v_j$ and $\tilde{x} \in \pi(\tilde v_i, \tilde v_j)$, it then follows that $g(\tilde x) \in [f(v_i) - 3\beta - 4\eps, f(v_j) + 3\beta + 4\eps]$ and thus $|g(\tilde x) - f(x)| \le 3\beta + 5\eps$ for any $x \in \rg_f$. Hence we have that $\max_{x  \in \rg_f} |f(x) - g\circ\leftmap^\eps(x)| \le 3\beta + 5\eps$.

Symmetrically, we can take an $\eps$-subdivision of $\rg_g$ with nodes $U_\eps = \{ \tilde u_1, \ldots, \tilde u_M \}$, and construct a continuous map $\rightmap^\eps: \rg_g \to \rg_f$. Using the same argument as above, we have that $\max_{\tilde y \in \rg_g} |g(\tilde y) - f(\rightmap^\eps(\tilde y))| \le 3\beta + 5\eps$. 

We now bound $\dreeb_f(x,y) - \dreeb_g(\tilde x, \tilde y)$ for any $(x,\tilde x), (y,\tilde y) \in G(\leftmap^\eps,\rightmap^\eps)$. 
If $x \in v_iv_j$ and $\tilde x \in \pi(\tilde v_i, \tilde v_j)$ (i.e, $\tilde x = \leftmap^\eps(x)$), we let $w_i=v_i$ and $\tilde w_i=\tilde v_i$ and have $d_f(x,w_i)\leq\eps$ and $d_g(\tilde x,\tilde w_i)\leq 2\beta + 3\eps$ (as the height of path $\pi(\tilde v_i, \tilde v_j)$ is at most $2\beta + 3\eps$ as discussed earlier). If, on the other hand, $\tilde x \in\tilde u_i \tilde u_j$ and $x \in \pi(u_i, u_j)$ (i.e, $x = \rightmap^\eps(\tilde x)$), we let $w_i=u_i$ and $\tilde w_i=\tilde u_i$ and have $d_f(x,w_i)\leq 2\beta + 3\eps$ and $d_g(\tilde x,\tilde w_i)\leq\eps$. In either case, we have $d_f(x,w_i)+d_g(\tilde x,\tilde w_i)\leq 2\beta + 4\eps$. 
See the illustrations of both cases below. 
\begin{figure}[h]
\centering
\includegraphics[scale=.75]{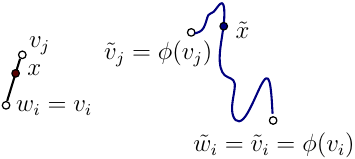}
\hfil
\includegraphics[scale=.75]{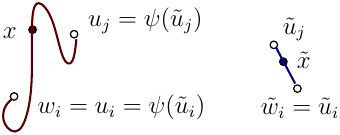}
\end{figure}
In an analogous way, we also obtain $w_a,\tilde w_a$ with $d_f(y,w_a)+d_g(\tilde y,\tilde w_a)\leq 2\beta + 4\eps$. 
Note that by the construction of $\leftmap^\eps$ and $\rightmap^\eps$, both $(w_i, \tilde w_i)$ and $(w_a, \tilde w_a)$ are from the $\eps$-optimal correspondence generated by the maps $(\leftmap, \rightmap)$. In other words, we have that $d_f(w_i, w_a) \le d_g(\tilde w_i, \tilde w_a) + 2\beta + 2\eps$. 
It then follows that:  
\begin{align*}
d_f(x,y) & \leq d_f(x,w_i)+d_f(w_i,w_a)+d_f(w_a,y) \\
& \leq d_f(x,w_i)+d_f(w_a,y)+ ( d_g(\tilde w_i,\tilde w_a)+2\beta+2\eps) \\
& \leq d_f(x,w_i)+d_f(w_a,y)+2\beta+2\eps+d_g(\tilde w_i,\tilde x)+d_g(\tilde x,\tilde y)+d_g(\tilde y,\tilde w_a) \\
& \leq d_g(\tilde x,\tilde y)+6\beta+10\eps.
\end{align*}
By symmetry of the above argument, we obtain \[\frac12\left|\dreeb_f(x,\tilde x)-\dreeb_g(y,\tilde y)\right| \leq 3\beta + 5 \eps.\]
Putting everything together, we  have that $\DD(\rg_f, \rg_g) \le \lim_{\eps \to 0} 3\beta + 5 \eps = 3d_{fGH}(\rg_f, \rg_g)$. 
\end{proof}

\section{Relation to Interleaving Distance for Merge Trees}
\label{sec:interleaving}

A merge tree is simply a rooted tree $T_f$ equipped with a function $f: T_f \to \reals$ such that the function value of $f$ from the root to any leaf is monotonically descreasing. 
For technical reasons, the version of ``merge trees" defined by Morozov et al. \cite{DMP14} further adds an extra arc from the root whose function value extends to $+\infty$, and they 
proposed an \emph{interleaving distance} for two merge trees under this modification. 
From now on, we assume merges are such extended merge trees. 
%

We first introduce the interleaving distance for merge trees defined in \cite{DMP14}. 
Assume that we are given two merge trees $\treeone$ and $\treetwo$ with associated functions $f: \treeone \to \reals$ and $g: \treetwo \to \reals$. 
\begin{mydef}[\cite{DMP14}]
Two continuous maps $\leftintmap: \treeone \to \treetwo$ and $\rightintmap: \treetwo \to \treeone$ are said to be \emph{$\eps$-compatible} for some $\eps \ge 0$, if 
\begin{align}
g(\leftintmap(x)) = f(x) + \eps~; ~~~~~~~~&~~~~~~~~ f(\rightintmap(y)) = g(y) + \eps; \label{eqn:shift1}\\
\rightintmap \circ \leftintmap = \ishift^{2\eps}~; 	~~~~~~~~&~~~~~~~~ \leftintmap \circ \rightintmap = \jshift^{2\eps}; \label{eqn:shift2}
\end{align}
where $\ishift^{2\eps}: \treeone \to \treeone$ and $\jshift^{2\eps}: \treetwo \to \treetwo$ are the $2\eps$-shift maps in the respective trees. 

The \emph{interleaving distance}, $\Dint(T_f, T_g)$, between two merge trees $T_f$ and $T_g$, is the greatest lower bound on $\eps$ for which there are $\eps$-compatible maps: 
\begin{equation}\label{eqn:intdistance}
\Dint(T_f, T_g) = \inf \{ \eps \mid ~\text{there are }~\eps~\text{compatible maps}~\leftintmap: \treeone \to \treetwo~\text{and}~\rightintmap: \treetwo \to \treeone \}. 
\end{equation}
\end{mydef}

Let $\DD(T_f, T_g)$ be the \GHlike{} distance for Reeb graphs that we introduced. 
The main result of this section is that, for merge trees, the interleaving distance of \cite{DMP14} and our \GHlike{} distance are isometric.
\begin{theorem}
Given two merge trees $T_f$ and $T_g$, equipped with functions $: \treeone \to \reals$ and $g: \treetwo \to \reals$, we have $\Dint(T_f, T_g) = \DD(T_f, T_g). $
\label{thm:int-FD}
\end{theorem}

\begin{proof}
We break down the proof into two steps, which are shown in \cref{lem:int-FD-1,lem:int-FD-2}.
\end{proof}

\begin{lemma}
\label{lem:int-FD-1}
$\Dfgh(T_f, T_g) \le \Dint(T_f, T_g)$.
\end{lemma}

\begin{proof}
Let $\eps = \Dint(T_f, T_g)$. We assume that $\eps$ is obtained by a pair of $\eps$-compatible maps\footnote{We note that $\eps$ may be only achieved in the limit. Our argument can be extended to that case by taking a sequence of $\eps'$-compatible maps and send $\eps'$ to $\eps$.},  $\leftintmap: \treeone \to \treetwo$ and $\rightintmap: \treetwo \to \treeone$. 
We will show that the correspondances generated by these two maps $\leftintmap$ and $\rightintmap$ induce a distance distortion at most $\eps$. This implies that $\DD(T_f, T_g) \le \eps$. 
Specifically, let $G(\leftintmap, \rightintmap)$ and $D(\leftintmap, \rightintmap)$ as introduced in \cref{eqn:GD}. 
We now bound $D(\leftintmap, \rightintmap)$. 

Consider two pairs $(x_1, y_1), (x_2, y_2) \in G(\leftintmap, \rightintmap)$. 
we first aim to bound $|d_f(x_1, x_2) - d_g(y_1, y_2)|$ from above. 

Assume first that $y_1 = \leftintmap(x_1)$ and $y_2 = \leftintmap(x_2)$. 
Let $\pi_1$ be the optimal path connecting $x_1$ to $x_2$ that achieves $d_f(x_1, x_2)$, which is necessarily the unique simple path connecting $x_1$ to $x_2$ in the tree $T_f$. 
Its image $\pi'_1 = \leftintmap(\pi_1)$ is a path connecting $y_1$ and $y_2$. 
By \cref{eqn:shift1}, $\leftintmap$ shift every point up by $\eps$ in the corresponding function value. Hence the range of $\pi_1$ is shifted up by $\eps$ to the range of $\pi'_1$ while their heights are the same. 
Hence we have $d_g(y_1, y_2) \le d_f(x_1, x_2)$.  

Now consider the optimal path $\pi_2$ connecting $y_1$ to $y_2$ to achieve $d_g(y_1, y_2)$ in $\treetwo$. 
Let $x'_1 = \rightintmap(y_1)$, $x'_2 = \rightintmap(y_2)$. The image $\pi'_2 = \rightintmap(\pi_2)$ of $\pi_2$ under the map $\rightintmap$ is a path connecting $x'_1$ to $x'_2$ in $\treeone$. 
Similarly, we have that $\height(\pi_2) = \height(\pi'_2)$ and the range of $\pi_2$ is translated up by $\eps$ to $\pi'_2$. 
On the other hand, by \cref{eqn:shift2}, we have $x'_1 = \ishift^{2\eps}(x_1)$, and $x'_2 = \ishift^{2\eps}(x_2)$. By the definition of the shift map, there is a monotone path from $x_1$ to $x'_1$ (along the path from $x_1$ to the root of the merge tree $\treeone$) in $\treeone$; and similarly for $x_2$ and $x'_2$. 
Concatenating these two montone paths with $\pi'_2$ we obtain a path $\pi_3$ connecting $x_1$ to $x_2$. 
Since the two new paths are monotone, of height $2\eps$ each, and both going up, we have that $\height(\pi_3) \le \height(\pi'_2) + 2\eps = \height(\pi_2) + 2\eps$. It then follos that $d_f(x_1, x_2) \le d_g(y_1, y_2) + 2\eps$. 
Putting this together with that $d_g(y_1, y_2) \le d_f(x_1, x_2)$ proved earlier, we thus have $|d_f(x_1,x_2) - d_g(y_1,y_2)| \le 2\eps$. 

If the two pairs are obtained via $x_1 = \rightintmap(y_1)$ and $x_2 = \rightintmap(y_2)$, a symmetric argument will show $|d_f(x_1,x_2) - d_g(y_1,y_2)| \le 2\eps$ as well. 

We now consider the remaining case where $y_1 = \leftintmap(x_1)$ but $x_2 = \rightintmap(y_2)$. 
Let $\pi$ be the optimal path connecting $x_1$ to $x_2$ in $\treeone$ to achieve $d_f(x_1,x_2)$. 
Let $\pi' = \rightintmap(\pi)$ be its image in $\treetwo$: note $\pi'$ connects $y_1$ to $y'_2 = \rightintmap(x_2)$. 
By \cref{eqn:shift1} of the definition of $\eps$-compatible maps, we have that $\pi'$ is of the same height of $\pi$ (and its range is that of $\pi$ shifted upward by $\eps$). 
By \cref{eqn:shift2} of the definition of $\eps$-compatible maps, we have that $y'_2 = \jshift^{2\eps}(y_2)$ and thus there is a monotone path $\pi_4$ of height $2\eps$ connecting $y_2$ to $y'_2$. 
Hence the concatenation $\pi_5 = \pi' \circ \pi_4$ is a path connecting $y_1$ to $y_2$. 
Thus $\height(\pi_5) \le \height(\pi') + 2\eps = \height(\pi) + 2\eps$, implying that $d_g(y_1, y_2) \le d_f(x_1,x_2) + 2\eps$. 

A symmetric argument shows that $d_f(x_1,x_2) \le d_g(y_1,y_2) + 2\eps$. Hence $|d_f(x_1,x_2) - d_g(y_1,y_2)| \le 2\eps$. It then follows that $D(\leftintmap, \rightintmap) \le \eps$. 
On the other hand, by \cref{eqn:shift1}, $\|f - g\circ \leftintmap\|_\infty = \eps$ and $\|f\circ\rightmapit - g\|_\infty = \eps$. 
By \cref{eqn:distdef}, it then follows that $\DD(T_f, T_g) \le \eps$. 
\end{proof}

\begin{lemma}
\label{lem:int-FD-2}
$\Dint(\treeone, \treetwo) \le \DD(\treeone, \treetwo)$. 
\end{lemma}
\begin{proof}
Let $\delta = \DD(\treeone, \treetwo)$ denote the \fGH-distance between two merge trees $\treeone$ and $\treetwo$, and 
let $\optleftmapit: \treeone\to \treetwo$ and $\optrightmapit: \treetwo \to \treeone$ be the optimal continuous maps\footnote{Again, if the optimal is achieved in the limit, we can modify our argument by taking a sequence of near optimal maps and take them to the limit.} achieving $\delta$. 
We will now construct a pair of $\delta$-compatible maps for $\treeone$ and $\treetwo$ using $\optleftmapit$ and $\optrightmapit$. 
This then implies that $\Dint(\treeone, \treetwo) \le \DD(\treeone,\treetwo)$ as claimed. 

First, we construct the map $\leftintmapdelta: \treeone\to \treetwo$ as follows: 
For every point $x\in \treeone$, let $y = \optleftmapit(x)$. Now set $\rho = f(x) + \delta - g\circ\optleftmapit(x)$ --- by the definition of $\DD$ in \cref{eqn:distdef}, $\rho$ is a non-negative real value in the range $[0, 2\delta]$. 
We now set $\leftintmapdelta(x) = \jshift^\rho (y) = \jshift^\rho \circ \optleftmapit (x)$. 
Easy to see that by the choice of $\rho$, $g(\leftintmapdelta(x)) = f(x) + \eps$. 
Since $\optleftmapit$ is continuous, the function $\rho: T_f \to \reals$ is continuous, and the map $\leftintmapdelta$ is thus also a continuous map. 
Similarly, we construct $\rightintmapdelta: \treetwo \to \treeone$. 
By their construction, the requirements in \cref{eqn:shift1} are satisfied. 
We now show that \cref{eqn:shift2} also hold for $\leftintmapdelta$ and $\rightintmapdelta$. 

Indeed, consider a point $x\in \treeone$, and let $y = \optleftmapit(x)$ and $y' = \leftintmapdelta(x)$. 
By the definitino of $\leftintmapdelta$, $g(y') = f(x) + \delta \ge g(y)$ and there is a monotone path $\pi$ connecting $y$ to $y'$ (in particular, $y'$ is along the path from $y$ to the root of the merge tree $\treetwo$). 
Now map $\pi$ back to $\treeone$ via the map $\rightintmapdelta$, which is necessarily a monotone path $\pi'$ connecting $\tilde{x} := \rightintmapdelta(y)$ and $x' := \rightintmapdelta(y') = \rightintmapdelta \circ \leftintmapdelta(x)$. In other words, $x'$ is along the path from $\tilde{x}$ to the root of the merge tree $\treeone$. 
By the definition of $\leftintmapdelta$ and $\rightintmapdelta$, $f(x') = f(x) + 2\delta$. 
We now show that $x'$ is along the path from $x$ to the root of the merge tree $\treeone$: this would then imply that $x' = \ishift^{2\delta}$, namely, $\rightintmapdelta \circ \leftintmapdelta = \ishift^{2\delta}$. 

To see that there is a monotone path from $x$ up to $x'$ in $\treeone$, set $\tilde{x}' = \optrightmapit(y)$. By the construction of $\rightintmapdelta$, $\tilde{x}$ is along the unique monotone path from $\tilde{x}'$ up to the root of $\treeone$. Furthermore, $f(\tilde{x}) = g(y) + \delta, f(\tilde{x}') \in [g(y) - \delta, g(y) + \delta]$ and $f(\tilde{x}') \le f(\tilde{x})$. 
Note that $(x, y)$ and $(\tilde{x}', y)$ are in the set of correspondances $\mathcal{G}(\optleftmapit,\optrightmapit)$ (this is because $y = \optleftmapit(x)$ and $\tilde{x}' = \optrightmapit(y)$). 
Hence by the definition of $\DD$ which is achieved by $\optleftmapit$ and $\optrightmapit$, there is a path $\tilde{\pi}$ connecting $x$ to $\tilde{x}'$ such that $\height(\tilde{\pi}) \le 2\delta$. This means that the least common ancester of $x$ and $\tilde{x}'$ has a function value at most $f(x) + 2\delta$ which is $f(x')$. 
Since $\tilde{x}'$ and $\tilde{x}$ are connected by a monotone path, the least common ancester of $x$ and $\tilde{x}$ has a function value at most $f(x')$. 
Since $x'$ is an ancester of $\tilde{x}$, it follows that $x'$ is an ancestor for $x$ as well. 
Hence $x' = \ishift^{2\delta}(x)$ and $\rightintmapdelta \circ \leftintmapdelta = \ishift^{2\delta}$. 
A symmetric argument shows that $\leftintmapdelta \circ \rightintmapdelta =\ishift^{2\delta}$. 
Putting everything together, we have that $\leftintmapdelta$ and $\rightintmapdelta$ form a $\delta$-compatible pair of maps between $\treeone$ and $\treetwo$. As such, $\Dint(\treeone, \treetwo) \le \delta = \DD(\treeone, \treetwo)$. 

\end{proof}


\section{Simplification of Reeb Graphs}
\label{sec:simp}

Reeb graphs have been used as a meaningful summary of the input functions. Simplifying a Reeb graph can help to remove noise or single out major features, and to create a multi-resolution representation of the input domain; see e.g., \cite{DN12,GSBW11,PSBM07}. 
As we described in \cref{sec:background}, there is a natural way to quantify branching and loop features in terms of ordinary and extended persistence in the according dimensions. 
Indeed, it is common practice to simplify the Reeb graph by removing all features with persistence smaller than a given threshold. In this section, we prove that by removing small features using a natural merging strategy, (branching and loop) features with large persistence will not be killed, and will roughly maintain their persistence (``importance''). 

\subsection{A Natural Simplification Scheme for Reeb Graphs}
\label{sec:simpstrategy}
\begin{figure}[t]
\begin{center}
\begin{tabular}{ccc}
\includegraphics[height=2.8cm]{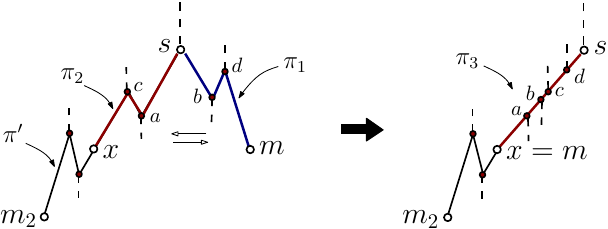} & &
\includegraphics[height=3.3cm]{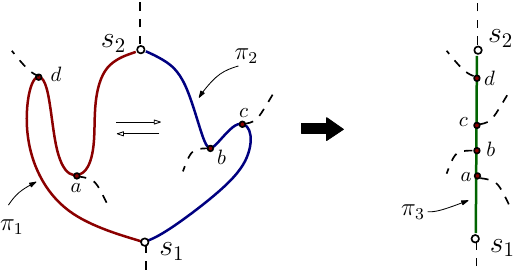}  \\
(a) & & (b)\end{tabular}
  \end{center}
\vspace*{-0.15in}\caption{{\small (a): Removing a branching feature spanned by $(m,s)$ by merging paths $\pi_1$ and $\pi_2$. This removes the point $(f(m), f(s))$ from the $0$-th ordinary persistence diagram. (b): Removing a cycle feature spanned by $(s_1,s_2)$; this removes the point $(f(s_1), f(s_2))$ from the $1$st extended persistence diagram.} 
\label{fig:merge}}
\end{figure}

We first introduce a natural simplification scheme for Reeb graphs (see, e.g., \cite{GSBW11,PSBM07}). See \cref{fig:merge} for an illustration.

Given an ordinary persistence pair $(m, s)$, assume that $m$ is a minimum and $s$ is a down-fork. 
Recall that the down-fork $s$ merges two connected components $C_1$ and $C_2$ of the sublevel set below $f(s)$, and $m$ is the higher minimum of the two. 
To remove the feature $(m, s)$, we wish to merge the branch containing $m$, say $C_1$, into the other branch $C_2$, so that afterwards, $m$ and $s$ become regular points (i.e., with up-degree and down-degree both being $1$). 
In particular, we perform the following operation (see \cref{fig:merge} (a)).  Let $m_2$ denote the minimum of $C_2$. We choose an arbitrary embedded path $\pi_1 \subseteq C_1$ from $s$ to $m$, and an arbitrary $\pi' \subseteq C_2$ from $s$ to $m_2$. Now imagine we traverse $\pi'$ starting from $s$. We stop when we encounter the first point $x$ on $\pi'$ such that $f(x) = f(m)$, and set $\pi_2$ to be the subcurve of $\pi'$ from $s$ to $x$. 
By identifying points with the same function value, we merge $\im\pi_1$ and $\im\pi_2$ to form the image of a new monotonic arc $\pi_3$ between $s$ and $x$ such that any point $p \in \im\pi_1 \cup \im\pi_2$ is mapped to some $q \in \im\pi_3$ with $f(q) = f(p)$. 
Pairs of type (up-fork, maximum) are treated in a symmetric way.

Given an extended persistence pair $(s_1, s_2)$ between an up-fork $s_1$ and a down-fork $s_2$, 
let $\gamma$ be a thin cycle that spans it. W.l.o.g. assume that $\im \gamma$ consists only of a single connected component: 
if $\im\gamma$ has multiple connected components, then there must exist one that contains both $s_1$ and $s_2$. That component is necessarily an embedded loop and thus we can simply set $\gamma$ to be the \mycanonical{} cycle corresponding to that loop. 
Let $\pi_1$ and $\pi_2$ denote the two disjoint sub-curves of the loop that connect $s_1$ and $s_2$. To cancel the feature, intuitively, we wish to merge $\pi_1$ and $\pi_2$ to kill the cycle $\gamma$. Note that $\pi_1$ and $\pi_2$ may not be monotonic (w.r.t. the input function $f$); however, all points in $\pi_1$ and $\pi_2$ have function values within the range $[f(s_1), f(s_2)]$. 
The merging of $\pi_1$ and $\pi_2$ results in a new monotonic arc $\pi_3$ from $s_1$ and $s_2$, such that every point $p \in \im\gamma$ is mapped to some $q \in \pi_3$ with $f(q) = f(p)$. See \cref{fig:merge} (b) for an illustration.

Note that since a critical pair $(m,s)$ (resp. an essential pair $(s_1,s_2)$) corresponds uniquely to a persistence pair $(f(m), f(s))$ in the ordinary persistence diagram (resp. $(f(s_1), f(s_2))$ in the extended persistence diagram), the above process also removes a point from the respective persistence diagram. 

Let $\onerg$ and $\onerg'$ denote the Reeb graph before and after the simplification of a persistence pair $\tau = (\bp, \tp)$ by collapsing its corresponding branching or loop feature. Let $\pi^\tau_1$ and $\pi^\tau_2$ be as introduced above. 
Call $\gamma^\tau = \pi_1^\tau \cup \pi_2^\tau$ the \emph{merging path w.r.t. $\tau$}. Note that $\gamma^\tau$ is a closed curve corresponding to a thin cycle spanning $(\bp, \tp)$ when it is an extended persistence pair, and a connected path with $\bp$ and $\tp$ being the respective minimum and maximum function values on it otherwise. 
In either case, the height of the merging path is at most $|d - b|$, the persistent of this pair $(b,d)$. 
The merging path $\gamma^\tau$ will be collapsed into a single monotonic arc in order to eliminate the persistence pair $\tau$. 
We can view the removal of $\tau$ in a more formal way as follows: We say that two points $x, y \in \onerg$ are $\tau$-equivalent, denoted  by $x \sim_\tau y$, if $f(x) = f(y)$ and $x, y \in \gamma^\tau$. 
The simplified Reeb graph $\onerg'$ is the quotient space $\onerg / {\sim}_\tau$; 
the corresponding quotient map $\mu_\tau: \onerg \rightarrow \onerg'$ satisfies $\mu_\tau(x) = \mu_\tau(y)$ if and only if $x \sim_\tau y$. 
The function $f: \onerg \rightarrow \reals$ induces a function $f': \onerg' \rightarrow \reals$ such that for any $x' \in \onerg'$, $f'(x') = f(x)$ for any $x \in \mu^{-1}_\tau (x')$. 

Now given an input Reeb graph $\onerg$, suppose we wish to eliminate a set of persistence pairs $\{ \tau_1 = (\bp_1, \tp_1), \tau_2 = (\bp_2, \tp_2), \ldots, \tau_k = (\bp_k, \tp_k) \}$. 
Compute the merging path $\gamma^{\tau_i}$ for each persistence pair $\tau_i$ in $\onerg$. 
We now define an equivalence relation $\sim$ as the transitive closure of all $\sim_{\tau_i}$s for $i \in [1, k]$. 
This is equivalent to collapsing $\gamma^{\tau_i}$s for all $i \in [1,k]$ in an arbitrary order to kill the persistence pairs $\tau_1, \ldots, \tau_k$. 
The final simplified Reeb graph $\newrg$ is obtained as the quotient space $\onerg / {\sim}$, with $\mu: \onerg \to \newrg$ being the associated quotient map. We have a well-defined function $g: \newrg \to \reals$ induced by the function $f: \onerg \to \reals$ such that $g(\mu(x)) = f(x)$ for any $x \in \onerg$. 
Let $\delta$ denote the largest persistence of $\tau_1, \ldots, \tau_k$. 
We have the following properties of $\newrg$. 

\begin{obs}
(i) Given any two points $x, y \in \onerg$, we have $d_g (\mu(x), \mu(y)) \le d_f(x, y)$. \\
(ii) Given a point $\tilde{x} \in \newrg$, for any two points $x_0, x_1 \in \mu^{-1}(\tilde{x})$ we have $d_f (x_0, x_1) \le 2\delta$. 
\label{obs:simpmap}
\end{obs}
\begin{proof}
Claim (i) follows easily since the quotient map $\mu$ preserves function values. We now prove (ii). 
Since $\mu(x_0) = \mu(x_1) = \tilde x$, by the definition of $\mu$ there exists a set of equivalent relations $\sim_{\tau_{j_1}}, \ldots, \sim_{\tau_{j_a}}$ with the index set $\{ j_1, \ldots, j_a \} \subseteq \{ 1, \ldots, k\}$ such that 
$y_0:= x_0 \sim_{\tau_{j_1}} y_1 \sim_{\tau_{j_2}} y_2 \cdots \sim_{\tau_{j_a}} y_a:= x_1$. 
Set $\alpha = f(x_0) = f(x_1)$. 
All~$y_i$ have the same function value $\alpha$. 
For each $i \in [1, a]$, we have that $y_{i-1} \sim_{\tau_{j_i}} y_i$, which is induced by the merging path $\gamma^{\tau_{j_i}}$ with $\height(\gamma^{\tau_{j_i}}) \le \delta$. In other words, there is a subpath $\pi_i$ of $\gamma^{\tau_{j_i}}$ connecting $y_{i-1}$ to $y_i$ such that $\range(\pi_i) \subseteq [\alpha-\delta, \alpha + \delta]$. 
The concatenation of these paths $\pi_i$ gives rise to a path $\pi$ connecting $y_0 = x_0$ and $y_a = x_1$, and $\range(\pi) \subseteq [\alpha-\delta, \alpha+\delta]$. 
This proves claim (ii). 
\end{proof}

A similar argument of the above observation can in fact lead to the following more refined statements.  
\begin{lemma}
Let $y_a, y_b \in \newrg$ be two points in $\newrg$ such that there exists a monotonic path $\pi^*$ between $y_a$ and $y_b$ with $d_g(y_a, y_b) =  \height(\pi^*) = g(y_b) - g(y_a)$, where $g(y_b) > g(y_a)$. 
Let $x_a$ and $x_b$ be arbitrary preimages for $y_a$ and $y_b$, respectively. Then $d_f(x_a,x_b) \le 2\delta + \height(\tilde{\pi})$. 

In fact, there is a path $\pi$ from $x_a$ to $x_b$ such that the highest point $t$ in $\im \pi$ satisfies $f(t) \le f(x_b) + \delta$, and the lowest point $w$ in $\im \pi$ satisfies $f(w) \ge f(x_a) - \delta$. 
\label{lem:ghpreimagepath} 
\end{lemma}

\subsection{Distance between $\onerg$ and $\newrg$}
\label{sec:simpdistance:GH}

While the simplification scheme removes persistence pairs $\tau_1, \ldots, \tau_k$, it is not clear how other points in the persistence diagram of the original Reeb graph $\onerg$ are affected. 
In this section, we aim to bound the \GHlike{} distance $\DD(\onerg, \newrg)$, which in turn will give an upper bound on the respective persistence diagrams. 
We do so through the functional Gromov-Hausdorff distance, $\Dfgh(\onerg, \newrg)$, between $\onerg$ and $\newrg$.
In particular,  by using the quotient  map $\mu: \onerg \to \newrg$ which describes the simplification process implemented on $\onerg$ so that $\newrg$ is obtained, we will show that the functional Gromov-Hausdorff distance between $\onerg$ and $\newrg$ is bounded by $\delta$. 

First, we rewrite the definition of functional GH distance in\cref{eqn:GHdistdef} by the following using the concept of correspondance: A correspondance $C \subset X \times Y$ between two topological spaces $X$ and $Y$ is a relation whose projection on $X$ and on $Y$ are both surjective. 
We can then rewrite \cref{eqn:GHdistdef} as follows: 
\begin{align}
D(C) &= \frac{1}{2} \max_{(x_1, y_1), (x_2, y_2) \in C} |d_f(x_1, x_2) - d_g (y_1, y_2) |; ~~\text{and} \nonumber \\
\Dfgh(\onerg, \newrg) &= \inf_%
{C: \onerg \times \newrg}
\max \{ D(C), \max_{(x,y) \in C} |f(x) - g(y)| \},  
\label{eqn:GHdistdef_reeb}
\end{align}
where $C$ ranges over all the correspondence between $\onerg$ and $\newrg$. 

Set $\widehat{C} = \{(x, \mu(x))|  x\in\onerg\}$. Note that this indeed is a correspondence since $\mu$ is a subjective map from $\onerg$ to $\newrg$. We will now bound $D(\widehat{C})$. 
Specifically, given any $\{(x_1, y_2), (x_2, y_2)\} \in \widehat{C}$ with $y_1 = \mu(x_1)$ and $y_2 = \mu(x_2)$, we aim to show that 
$|d_f(x_1,x_2) - d_g(y_1,y_2)| \leq 2\delta$; that is, 
\begin{align}
-2\delta \leq d_f(x_1,x_2) - d_g(y_1,y_2) \leq 2\delta
\label{eqn:ghdreeb}
\end{align}

To see that the left inequality in \cref{eqn:ghdreeb} holds, note that by \cref{obs:simpmap}, we have $ d_g( y_1 ,y_2) - d_f(x_1, x_2) \leq 0$. 

\begin{figure*}[htbp]
\begin{center}
\includegraphics[height=4.5cm]{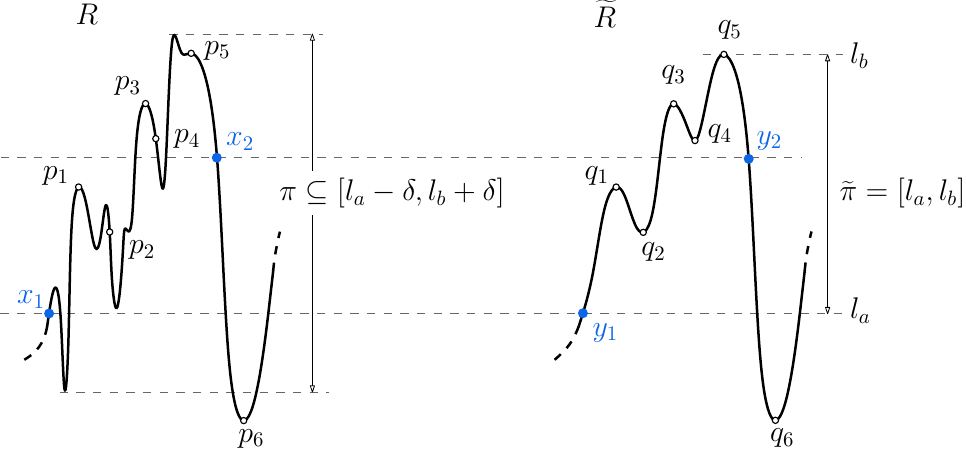}
\end{center}
\caption{Right: An arc $(q_i, q_{i+1}) \in \newrg$ denotes a monotonic path. Each of them has a correspondent path $(p_i, p_{i+1}) \in \onerg$, with $q_i = \leftmap(p_i)$. The path between $y_1$ and $y_2$ is the concatenation of a set of monotonic paths, i.e., $\tilde{\pi}(y_1, y_2) = \lbrace (y_1, q_1,), \cdots, (q_5, y_2) \rbrace \in \newrg $ with  $\range(\tilde{\pi}(y_1, y_2)) = [l_a, l_b]$. Left: There exists a path $\pi(x_1, x_2)$ in $\onerg$, with $x_1$ and $x_2$ be the arbitrarily preimages of $y_1$ and $y_2$, respectively. The range of $\pi(x_1, x_2)$ can be bounded as $[l_a-\delta, l_b+\delta]$. } 
\label{fig:ghreeb_right} 
\end{figure*}

We now show the right part of \cref{eqn:ghdreeb}. 
Assume w.l.o.g that the Reeb graph $\onerg$ and thus also $\newrg$ are connected. 
Let $\tilde \pi (y_1, y_2)\in \newrg$ denote the path with the minimum height connecting $y_1$ and $y_2$ (i.e, achieving $d_g(y_1, y_2)$) in $\newrg$. 
Suppose that $\tilde{\pi}(y_1, y_2)\in \newrg$ is the concatenation of a set of monotonic paths in $\newrg$; see \cref{fig:ghreeb_right}: 
\begin{align*}
\tilde{\pi}(y_1, y_2) = \lbrace \tilde\pi(y_1=q_0, q_1), \tilde\pi(q_1, q_{2}), \cdots, \tilde\pi(q_{s-1}, q_s= y_2) \rbrace \in \newrg
\end{align*}
By \cref{lem:ghpreimagepath}, each $\tilde{\pi}(q_i, q_{i+1})$ gives rise to a path $\pi(p_i, p_{i+1}) \in \onerg$ such that $\mu(p_i) = q_i, \mu(p_{i+1}) = q_{i+1}$, and $\height(\pi^*(p_i,p_{i+1})) \le 2\delta + \height(\tilde{\pi}^*(q_i, q_{i+1}))$. In fact, we can choose $p_0$ and $p_s$ as $x_1$ and $x_2$ (which are preimages of $y_1$ and $y_2$), respectively. 
Concatenating all $\pi(p_i, p_{i+1})$, for $i=0,\cdots,s$, we obtain a path $\pi(p_0=x_1, p_s=x_2)$ with $\height(\pi(p_0, p_s)) \leq  \height(\tilde{\pi}(q_0, q_s)) + 2\delta$. Hence
 $$d_f(x_1, x_2) \le \height(\pi(x_1=p_0, x_2=p_s)) \leq \height( \tilde{\pi}(y_1=q_0, y_2=q_s)) + 2\delta = d_g(y_1, y_2) + 2\delta. $$ 
The right part of \cref{eqn:ghdreeb} thus  holds. Hence $D(\widehat{C}) \le \delta$. 

Furthermore, since $\max_{(x,y) \in \widehat{C}} |f(x) - g(y)| = 0$ (as for any $y=\mu(x)$, $g\circ \mu(x) = f(x)$), we have that $\Dfgh(\onerg, \newrg) \le \delta$. 
Therefore, by \cref{thm:GHrelations}, we have 
\begin{align}
d_{FD}(\onerg, \newrg)\leq 3d_{fGH}(\onerg, \newrg) = 3\delta
\end{align}

Combining this with \cref{thm:traditional,thm:extendedbound}, we conclude with the following main result on the simplification of the Reeb graphs: 

\begin{theorem}
Suppose we simplify a Reeb graph $\onerg$ by removing features of persistence $\leq \delta$ using the strategy detailed in \cref{sec:simpstrategy}. 
The bottleneck distance between the (ordinary and extended) persistence diagrams for $\onerg$ and for its simplification $\newrg$ is at most $9\delta$. 

\label{thm:simp}
\end{theorem} 

We remark that
instead of invoking \cref{thm:GHrelations}, one can use a direct argument similar to the proof of that theorem to improve the bound on $\DD(\onerg, \newrg)$ to $2\delta$, which further improves the bound on bottleneck distance between persistence diagrams for $\onerg$ and $\newrg$ to $6\delta$.

\section{Concluding Remarks}
In this paper, we propose a distance for Reeb graphs, under which the Reeb graph is stable with respect to changes in the input function under the $L_\infty$ norm. More importantly, we show that this distance is bounded from below by and thus more discriminative at differentiating scalar fields than the bottleneck distance between both 0th ordinary and 1st extended persistence diagrams. 
Similar to the Gromov-Hausdorff distance for metric spaces, the functional distortion distance provides a rigorous setting for describing and studying various properties of Reeb graphs. Indeed, by bounding the \GHlike{} distance between a Reeb graph and its simplification, we can prove that important (persistent) features are preserved under simplification, which addresses a key practical issue.

Our current bound in Theorem \ref{thm:extendedbound} has a constant factor of $3$. It will be interesting to see whether this factor can be improved to $1$ to match the bound in Theorem \ref{thm:traditional}, either for the \GHlike{} distance or for some other distance.

A natural question is how to compute the \GHlike{} distance. %
We believe that there is an exponential time algorithm to approximate $\DD(\rg_f, \rg_g)$, similar to what is known for the $\eps$-interleaving distance for merge trees \cite{MBW13}. 
However, it remains an open problem to develop more efficient algorithms. 
We remark that comparing unlabeled trees is computationally hard in general: The commonly used tree edit distance and tree alignment distance 
are NP-hard to compute (and sometimes even to approximate) \cite{B05}. 
Similarly, it has been shown that computing the Gromov-Hausdorff distance is NP-hard even for two metric trees \cite{AFN15}. 
It will be interesting to see whether by leveraging the scalar field associated with merge trees and Reeb graphs, more efficient approximation algorithms for computing \GHlike{} distance can be developed.

\subsection*{Acknowledgements}
We thank Facundo M\'emoli for helpful discussions about variants of the Gromov--Hausdorff distance, which lead to improvements in our definition of the functional distortion distance. 
This research is partially supported by the National Science Foundation under grants CCF-1319406, CCF-1116258, and by the {\sc Toposys} project FP7-ICT-318493-STREP.

{
\bibliography{refs}
}

\todos

\end{document}